%% file: main.tex
\documentclass[a4paper,UKenglish,cleveref, autoref, thm-restate, cleveref]{lipics-v2021}

\hideLIPIcs  

\usepackage{amsmath,amssymb,mathtools,amsthm,braket}

\theoremstyle{plain}
\newtheorem{result}{Result}

\usepackage[ruled,section]{algorithm}
\usepackage{tabularx, multirow, booktabs}

\usepackage{algpseudocode}

\floatname{algorithm}{Protocol}
\crefname{algorithm}{protocol}{protocols}
\Crefname{algorithm}{Protocol}{Protocols}

\DeclareMathOperator{\Tr}{Tr}
\DeclareMathOperator{\SR}{SR}

\DeclareMathOperator{\poly}{poly}
\newcommand{\QMA}{\mathsf{QMA}}
\newcommand{\QMALOCC}{\mathsf{QMA}^\mathrm{LOCC}}
\newcommand{\BQP}{\mathsf{BQP}}
\newcommand{\NP}{\mathsf{NP}}

\newcommand{\NEXP}{\mathsf{NEXP}}

\newcommand{\id}{\mathbb I}
\newcommand{\ip}[2]{\langle #1\mid #2\rangle}
\newcommand{\proj}[1]{\ket{#1}\!\bra{#1}}
\newcommand{\norm}[1]{\lVert #1\rVert}
\newcommand{\abs}[1]{\lvert #1\rvert}
\newcommand{\dd}{\mathcal D}
\newcommand{\llin}{\mathcal L}

\newtheorem*{scoreboundsrestated}
  {Lemma~\ref{lem:four-witness-score}}
  \newtheorem*{paddinglemmarestated}{Lemma~\ref{thm:padding}}
\newtheorem*{locctheoremrestated}
{Theorem~\ref{prop:sn-padding}}

\algrenewcommand{\algorithmicrequire}{\textbf{Registers:}}
\title{QMA(2) with Limited Shared Entanglement} 

\author{Alex {Della Schiava}}{Graduate School of Mathematics, Nagoya University, Japan \and \url{https://alex-dell1.github.io}}{alex.della.schiava.d2@math.nagoya-u.ac.jp}{https://orcid.org/0009-0007-6022-4961}{JST SPRING, Grant Number JPMJSP2125.}

\author{Ranitha Mataraarachchi}{Graduate School of Mathematics, Nagoya University, Japan \and \url{https://ranitha.com}}{ranitha@nagoya-u.jp}{https://orcid.org/0009-0003-6053-6667}{JST CREST, Grant Number JPMJCR24I4.}

\authorrunning{A. Della Schiava and R. Mataraarachchi} 

\Copyright{Alex Della Schiava and Ranitha Mataraarachchi} 



\category{} 

\relatedversion{} 

\acknowledgements{The authors thank François Le Gall for helpful discussions and Seiichiro Tani for pointing out the formulation of this question in~\cite{ABD09}.}

\nolinenumbers 

\EventEditors{John Q. Open and Joan R. Access}
\EventNoEds{2}
\EventLongTitle{42nd Conference on Very Important Topics (CVIT 2016)}
\EventShortTitle{CVIT 2016}
\EventAcronym{CVIT}
\EventYear{2016}
\EventDate{December 24--27, 2016}
\EventLocation{Little Whinging, United Kingdom}
\EventLogo{}
\SeriesVolume{42}
\ArticleNo{23}

\begin{document}

\maketitle

\begin{abstract}
In \(\QMA(2)\), two computationally unbounded provers submit
unentangled quantum witnesses to a polynomial-time quantum
verifier. We ask whether allowing limited shared entanglement between the two provers changes the power of \(\QMA(2)\). To this end, we formalize \(\QMA(2;h)\), a variant
in which the provers initially share \(h\) EPR pairs, as suggested by Aaronson, Beigi, Drucker, Fefferman, and Shor (ToC~2009).
We show that polynomially many shared EPR pairs give the
provers no cheating advantage: \(\QMA(2)\subseteq\QMA(2;h)\) for \(h=\operatorname{poly}(n)\).
For logarithmically many pairs, the power of \(\QMA(2)\)
remains unchanged, strengthening this inclusion to
\(\QMA(2)=\QMA(2;h)\) for \(h=O(\log n)\).
Finally, we introduce a variant that also allows classical communication between the provers  when preparing the witnesses. This setting is particularly revealing for logarithmic-size
witnesses: assuming \(\NP\nsubseteq\BQP\), we show that superlogarithmically many shared EPR pairs strictly
reduce the model's computational power.
\end{abstract}

\clearpage

\input{sections/01_introduction}
\input{sections/02_Model_definition}

\input{sections/03_monotonicity}

\input{sections/04_removing_log_entanglement}
\input{sections/05_LOCC}

\input{sections/06_logarithmic_witnesses}

\clearpage

\bibliography{references}

\clearpage

\appendix
\numberwithin{theorem}{section}
\numberwithin{definition}{section}
\numberwithin{lemma}{section}
\numberwithin{proposition}{section}
\numberwithin{corollary}{section}

\input{appendix/A_preliminaries}
\input{appendix/B_monotonicity}
\input{appendix/C_removing_entanglement}

\input{appendix/D_locc}

\input{appendix/E_witness_grouping}

\end{document}

%% file: sections/01_introduction.tex
\section{Introduction}

In \(\QMA(2)\), two computationally unbounded provers submit unentangled quantum witnesses to a polynomial-time quantum verifier~\cite{KMY03}. Harrow and Montanaro showed that two
witnesses always suffice: \(\QMA(2)=\QMA(k)\) for every polynomially bounded
\(k\ge 2\)~\cite{HM13}. Nevertheless, the precise power of \(\mathsf{QMA}(2)\) remains poorly understood, with the known bounds given by \(\QMA \subseteq \QMA(2) \subseteq \NEXP\).

The possible advantage of \(\QMA(2)\) over \(\QMA\) rests
on its unentanglement promise, which restricts both
the witnesses available to honest provers and the dishonest cheating strategies. Recent evidence shows that the restriction can confer a computational advantage: relative to an oracle, \(\QMA(2)\) is strictly more powerful than \(\QMA\)~\cite{BGH26, GMT26}. Beyond the comparison with \(\QMA\), the recent survey by Jeronimo, Leigh, and Wu~\cite{JLW} outlines several approaches to characterizing the computational power of \(\QMA(2)\). One of these examines how additional structural restrictions on the witnesses affect computational power~\cite{JW23,BFM24,JW24,BFLMW24}. For example, requiring nonnegative amplitudes~\cite{JW23} or internal separability~\cite{BFLMW24} yields variants of \(\QMA(2)\) that characterize \(\NEXP\) for suitable constant completeness and soundness parameters.

Our work pursues a complementary direction. Rather than impose additional restrictions, we relax the unentanglement promise of \(\mathsf{QMA}(2)\) and ask:
\begin{quote}
\emph{Does limited shared entanglement between provers
affect the power of \(\mathsf{QMA}(2)\)?}
\end{quote}
The question goes back to Aaronson, Beigi, Drucker, Fefferman, and Shor~\cite{ABD09}, who suggested $\mathsf{QMA}(2;h)$, a variant in which the provers share $h$ EPR pairs before preparing their witnesses. They asked whether, for any fixed polynomial~\(h\), \(\mathsf{QMA}(2;h)=\mathsf{QMA}(2)\).

Our paper formalizes $\mathsf{QMA}(2;h)$: starting from $h$ shared EPR pairs and local workspace, each prover applies an arbitrary local quantum channel and submits the resulting state to the verifier. Under this preparation rule, the structural properties of a legal witness depend on whether it is pure or mixed:
\begin{itemize}
    \item \textbf{Pure witnesses:} Equal nonzero Schmidt coefficients;
    \item \textbf{Mixed witnesses:} Schmidt number at most $2^h$.
\end{itemize}

Unlike in \(\QMA(2)\), convexity does not directly justify restricting to pure witnesses in \(\QMA(2;h)\) because the pure components of a legal mixed witness need not themselves be legal. This observation led us to relax the witness set further, allowing any witness of Schmidt number at most \(2^h\). This setting has an exact operational interpretation: aside from \(h\) shared EPR pairs and local operations, the provers are now endowed with classical communication. We denote the model accordingly as \(\QMALOCC(2;h)\).

A variant of $\QMA(2)$ restricts each witness to logarithmically many qubits: $\QMA_{\log}(2)$. While a single logarithmic-size witness yields only $\BQP$~\cite{MW05}, allowing an inverse-polynomial completeness--soundness gap gives $\NP \subseteq \QMA_{\log}(2)$~\cite{BT12}.
Given the quantitative nature of our central question, it is natural to ask whether the answer changes when the witnesses are restricted to logarithmic size. To this end, we extend the notation to the shared EPR setting, writing \(\QMA_{\log}(2;h)\) and \(\QMALOCC_{\log}(2;h)\).

\subsection{Our results} Our first result shows that allowing more shared EPR pairs does not decrease the power of \(\QMA(2)\), provided their number remains polynomial.
\begin{result}[\Cref{cor:monotonicity}]\label{res:one}
    For any polynomially bounded \(h \leq H\),
    \(
        \QMA(2;h)\subseteq\QMA(2;H).
    \)
\end{result}
Our model thus stands out from prior work~\cite{JW23,BFM24,JW24,BFLMW24}, as we retain the full power of \(\QMA(2)\) despite a weaker witness promise. One alternative interpretation of \Cref{res:one} is that polynomially many shared EPR pairs give dishonest provers no additional cheating power.

A simple corollary of \Cref{res:one} advances the original question posed in~\cite{ABD09}: Do polynomially many EPR pairs change the power of \(\QMA(2)\)?
\begin{result}[\Cref{cor:polynomial-budgets}]\label{res:two}
For any constant $\varepsilon>0$ and any polynomially bounded $h$,
\[
\QMA(2; n^\varepsilon)=\QMA(2)
\quad\Longrightarrow\quad
\QMA(2;h)=\QMA(2).
\]
\end{result}
That is, to establish \(\mathsf{QMA}(2;h)=\mathsf{QMA}(2)\) for every polynomially bounded \(h\), it suffices to prove the equality for even just a very small polynomial \(h=n^\varepsilon\) with some constant \(\varepsilon>0\).

Our main result establishes this equality for logarithmically many shared EPR pairs.
\begin{result}[\Cref{thm:logarithmic-entanglement}]\label{res:three}
    For any \(h = O(\log n)\),    \(
        \QMA(2;h) = \QMA(2).
    \)
\end{result}

All class containments and equalities stated thus far preserve constant completeness--soundness gaps and extend to the \(\QMALOCC(2;h)\) setting. With an inverse-polynomial gap, \Cref{res:three} also extends to logarithmic-size witnesses, where the value of the LOCC setting becomes clearest. Within this domain, we show that sharing more than logarithmically many EPR pairs strictly reduces the model’s computational power, assuming \(\NP \nsubseteq \BQP\).
\begin{result}[\Cref{cor:log-superlogarithmic-barrier}]\label{res:four}
    Assuming $\NP\nsubseteq\BQP$, for every  $H=\omega(\log n)$,
\[
\QMA_{\log}^{\mathrm{LOCC}}(2;H)
\subsetneq \QMA_{\log}(2).
\]
\end{result}
As the number of shared EPR pairs exceeds the witness size,
the provers can prepare any state on the joint registers.
The model therefore collapses to
\(\QMA_{\log}=\BQP\)~\cite{MW05}, yielding a strict loss
of power under the stated assumption.

\subsection{Our techniques} For ease of exposition, we sketch the main technical ideas in the pure-state setting.\footnote{The full proofs handle mixed witnesses through use of purifications.} As mentioned above, pure witnesses in the \(\QMA(2;h)\) model have equal nonzero Schmidt coefficients. These states are called \emph{flat} and can be written as
\begin{equation}\label{eq:flat_def}
    \ket{\psi}_{AB}
    = \frac{1}{\sqrt{R}}\sum_{i=1}^{R}
      \ket{a_i}_A\ket{b_i}_B, \qquad \text{where } \quad R=2^t,\quad t\in\{0,\ldots,h\}.
\end{equation}
Here \(A\) and \(B\) denote the witness registers submitted by the two provers.

\subsubsection{Result 1: Testing surplus EPR pairs}

Let $h\leq H$. Our goal is to adapt a $\QMA(2;h)$ verifier so that it remains sound against $H-h$ additional EPR pairs. To this end, we ask the honest provers to submit a legal original witness $\ket{\psi}_{AB}$ alongside the surplus EPR registers \(C_AC_B\), left untouched:
\[
\ket{\psi}_{AB}\otimes \ket{\Phi^+}^{\otimes(H-h)}_{C_A C_B}.
\]
The adapted verifier tests all surplus Bell pairs through a Bell measurement, rejecting unless every outcome corresponds to \(\ket{\Phi^+}\). If all tests pass, it simulates the original verifier on \(AB\). Completeness of the protocol follows immediately. For soundness, we show that tampering with the surplus EPR pairs cannot increase the verifier's acceptance probability.

\subsubsection{Result 3: Four-witness simulation}\label{subsec:intro_four_witness}

For $h=O(\log n)$, our goal is to show \(\QMA(2) = \QMA(2;h)\). \Cref{res:one} immediately gives
$\QMA(2)\subseteq\QMA(2;h)$ by increasing the number
of shared EPR pairs from zero to \(h\). For the reverse inclusion, we first simulate $\QMA(2;h)$ using four unentangled witnesses, then apply the Harrow--Montanaro equality $\QMA(4)=\QMA(2)$~\cite{HM13}.

To see why we rely on four witnesses, observe that the main
difficulty in simulating $\QMA(2;h)$ lies in the \emph{flat} structure of
its pure witnesses, as shown for \(\ket{\psi}_{AB}\) in \eqref{eq:flat_def}. Our \(\QMA(4)\) construction enforces this structure by means of \emph{two} witnesses; each witness encodes one of the local
Schmidt bases, pairing every basis vector with its index:
\begin{equation}\label{eq:proper_encodings_intro}
\ket{\alpha}_{AI}
= \frac{1}{\sqrt R}\sum_{i=1}^{R}
  \ket{a_i}_A\ket{i}_I,
\qquad
\ket{\beta}_{BJ}
= \frac{1}{\sqrt R}\sum_{i=1}^{R}
  \ket{b_i}_B\ket{i}_J.
\end{equation}
We call \(\ket{\alpha}_{AI} \otimes \ket{\beta}_{BJ}\) a
\emph{proper encoding} of \(\ket{\psi}_{AB}\).
The two additional witnesses, on registers \(A'I'\) and \(B'J'\),
are intended as copies, which the verifier will use to test whether the received encodings are indeed proper.

More concretely, the verifier's toolbox includes three types of tests:
\begin{enumerate}
    \item \textbf{Consistency test:} SWAP test between two copies (e.g., \(AI\) and \(A'I'\)).\\ Checks that the copies agree with each other.
    \item \textbf{Properness test:} SWAP test between the indexes of two copies (e.g., \(I\) and \(I'\)).\\
    If the witnesses on \(AI\) and \(A'I'\) (respectively, \(BJ\) and \(B'J'\)) are identical, the SWAP test estimates the purity of \(I\), which is minimized when \(I\) is maximally mixed. From~\eqref{eq:proper_encodings_intro}, we see that an encoding is proper when the index is maximally mixed. Then, the SWAP test effectively tests \emph{properness} of \(AI\) (respectively \(BJ\)).
    \item \textbf{Bell test:} Bell measurement on indexes \(IJ\).\\
    This Bell measurement does not serve as a test; its purpose is to reconstruct the original witness \(\ket{\psi}_{AB}\) from the two encodings  \(\ket{\alpha}_{AI} \otimes \ket{\beta}_{BJ}\). On proper encodings, the Bell projection on \(IJ\) succeeds with probability \(R^{-2}\), and pairs the basis vectors with matching indices, leaving \(AB\) in the state
\[
\ket{\psi}_{AB}
= \frac{1}{\sqrt R}\sum_{i=1}^{R}
  \ket{a_i}_A\ket{b_i}_B.
\] 
\end{enumerate}

The full protocol is given in \Cref{prot:four-witness-simulation}
in Section~4. On each run, the verifier performs at most one of
these tests, chosen at random. A suitable choice of the test
probabilities yields a completeness--soundness gap of
\begin{equation}\label{eq:intro_gap}
    \Omega\left(\frac{1}{\operatorname{poly}(n) \,D^2}\right), \qquad D = 2^h.
\end{equation}
In particular, the gap is inverse polynomial for logarithmically many EPR pairs \(h=O(\log n)\). Harrow--Montanaro amplification~\cite{HM13} then yields a \(\QMA(2)\) protocol with constant gap.

\subsection{Shared entanglement beyond QMA(2)}
Kobayashi and Matsumoto~\cite{KM03} showed that quantum
multiprover interactive proof systems with polynomially many
prior-entangled qubits lie in \(\NEXP\).
Their upper bound also applies to our setting, yielding
\(\QMA(2;h)\subseteq\NEXP\) for polynomially bounded \(h\).

A related ``interactive analogue'' of our main question was studied by Gavinsky~\cite{Gav08}, who showed that classical two-prover interactive proofs retain their $\NEXP$ power under any polynomial bound on shared entanglement fixed in advance.

Ito and Vidick~\cite{IV12} subsequently constructed three-prover interactive proofs for \(\NEXP\) that remain sound against arbitrarily entangled dishonest provers.

\subsection{Discussion and open problems}

We now discuss the scope of our work. For each of our results we outline the questions they leave open and possible directions for further research.
\begin{itemize}
    \item \textbf{Path to \(\NEXP\).} Our work investigates $\QMA(2)$ by relaxing its unentanglement promise. \Cref{tab:summary} summarizes our results alongside prior work, which instead imposes additional restrictions on the witnesses. While \Cref{res:one} shows that \(\QMA(2;h)\) retains the power of \(\QMA(2)\), we do not yet know a configuration for which it captures \(\NEXP\). For polynomially bounded \(h\), we have
    \[
    \QMA(2)\subseteq\QMA(2;h)\subseteq\NEXP.
    \]
    A natural question then is: can polynomially many shared EPR pairs capture \(\NEXP\)?
    \item \textbf{Path to \(\QMA(2)\).}
Closely related is the question motivating our work,
originally posed in~\cite{ABD09}: do polynomially many shared
EPR pairs increase the power of \(\QMA(2)\) at all?
Results~\ref{res:two} and~\ref{res:three} make progress
on this question, but do not settle it. The limitation of our techniques is apparent
from~\eqref{eq:intro_gap}, since the factor \(D^{-2}\)
makes the guaranteed gap exponentially small for \(h = \operatorname{poly}(n)\). This loss is inherent in the Bell step, whose success probability
can be as low as \(D^{-2}\) even on honest proper encodings.
    \item \textbf{The power of flat witnesses.} \Cref{res:four} gives an entanglement threshold relative to
the witness size, beyond which \(\QMALOCC_{\log}(2;h)\)
collapses to the single-witness model \(\QMA_{\log}\). Under suitable complexity assumptions, this model is strictly weaker.\footnote{For logarithmic-size witnesses, we use $\NP\nsubseteq\BQP$. For polynomial-size witnesses,
extending the model to superpolynomially many shared EPR pairs
yields the analogous conclusion under $\QMA\neq\QMA(2)$.} Does the same collapse occur in the \(\mathsf{QMA}(2;h)\) model? Even with arbitrarily many shared EPR pairs, pure witnesses are restricted to flat states. Let \(\QMA^{\mathrm{flat}}\) be a variant of \(\QMA\) in which
the witness is promised to be a pure flat state across a fixed
bipartition. Is \(\QMA^{\mathrm{flat}}=\QMA\)?
\end{itemize} 

\begin{table}[t]
\centering

\begin{threeparttable}
\caption{The power of \(\QMA(2)\) with restricted or relaxed
witness sets.}
\label{tab:summary}
\begin{tabular*}{\textwidth}{
  @{\extracolsep{\fill}}
  l l l r
  @{}
}
\toprule
\multicolumn{2}{@{}l}{\textbf{Complexity class}}
 & \textbf{Result (constant gap)}
 & \textbf{Reference} \\
\midrule

\multicolumn{2}{@{}l}{$\mathsf{QMA}^{+}(2)$:  Nonnegative amplitudes}
 & $\mathsf{QMA}^{+}(2)=\mathsf{NEXP}$
 & \cite{JW23} \\

\multicolumn{2}{@{}l}{$\mathsf{QMA}_{\mathrm{IS}}(2)$: Internally separable}
 & $\mathsf{QMA}_{\mathrm{IS}}(2)=\mathsf{NEXP}$
 & \cite{BFLMW24} \\

\toprule
\multicolumn{4}{@{}l@{}}{
  $\mathsf{QMA}(2;h)$: Shared EPR pairs
} \\[3pt]

\multicolumn{2}{@{}l}{$h=\operatorname{poly}(n)$}
 & $\mathsf{QMA}(2)\subseteq\mathsf{QMA}(2;h)$
 & \Cref{res:one} \\

\multicolumn{2}{@{}l}{$h=O(\log n)$}
 & $\mathsf{QMA}(2;h)=\mathsf{QMA}(2)$
 & \Cref{res:three} \\

\toprule
\multicolumn{4}{@{}l@{}}{
  $\mathsf{QMA}^{\mathrm{LOCC}}_{\log}(2;h)$:
  {\footnotesize Shared EPR pairs + LOCC; Logarithmic-size witnesses.\tnote{1}}
} \\[3pt]

\multicolumn{2}{@{}l}{$h=O(\log n)$}
 & $\mathsf{QMA}^{\mathrm{LOCC}}_{\log}(2;h)
      =\mathsf{QMA}_{\log}(2)$
 & \Cref{res:three} \\

\multicolumn{2}{@{}l}{$h=\omega(\log n)$}
 & $\mathsf{QMA}^{\mathrm{LOCC}}_{\log}(2;h)
      =\mathsf{QMA}_{\log}$
 & \Cref{res:four} \\

\bottomrule
\end{tabular*}
\begin{tablenotes}[flushleft]
\item[1] In the $\QMA_{\log}$ setting,
gaps are inverse-polynomial.
\end{tablenotes}
\end{threeparttable}
\end{table}

\subsection{AI usage disclosure}
OpenAI's GPT-5.6 Sol and GPT-6 Astra (via ChatGPT), together
with Anthropic's Claude 5.0 Opus, contributed substantially
to both the research and the manuscript. The authors regard these contributions as comparable in scope
to those of a primary human contributor.

In August 2026, the authors provided GPT-5.6 Sol with their
own research notes, containing their formalization of
$\QMA(2;h)$ and characterization of its legal witnesses.
From this input, the model generated the four-witness
simulation used to establish \Cref{res:three}. Later interactions with GPT-6 Astra yielded a direct
two-witness simulation, bypassing $\QMA(4)$. Upon review, the authors retained the original
proof for clarity of exposition.
In September 2026, the authors sought to formalize a version
of the model encompassing all polynomially bounded numbers
of shared EPR pairs. Interactions with the LLMs during this
process led to \Cref{res:one} and, consequently, \Cref{res:two}. The authors then asked GPT-6 Astra to consider the case of
logarithmic-size witnesses, which led to \Cref{res:four}.

All listed models also assisted in preparing the
manuscript, which the authors revised in full. The authors take full responsibility for the
mathematical claims and the final text.

%% file: sections/02_Model_definition.tex
\section{QMA(2) with shared EPR pairs}
\label{sec:shared-epr}

We formalize the setting suggested in~\cite[Section~6.2]{ABD09}
to study the power of shared EPR pairs in $\QMA(2)$.
Preliminaries and notation are deferred to \Cref{sec:prelim}.

\subsection{Definition of the model}

Let $h$ be an efficiently computable, nonnegative polynomially bounded function. On input $x$ of length $n$, the two provers initially share $h=h(n)$ EPR pairs on registers $R_A,R_B$. They apply arbitrary local channels to prepare witnesses on $A,B$, potentially using private ancillas. The channels may depend on $x$. With $D=2^h$, the admissible witness set is
\begin{equation}\label{eq:witness-set}
\mathcal{W}_h(A:B):=\left\{
  (\Lambda_A\otimes\Lambda_B)(\proj{\Phi_D}_{R_AR_B})
  \;\middle|\;
  \begin{array}{l}
    \Lambda_A:\llin(R_A)\to\llin(A),\\
    \Lambda_B:\llin(R_B)\to\llin(B)
  \end{array}\text{ are channels}
\right\}.
\end{equation}

\begin{definition}\label{def:qmah}
A promise problem $L=(L_{\mathrm{yes}},L_{\mathrm{no}})$ is in $\QMA(2;h;a,b)$ if it has a uniform polynomial-time quantum verifier that, on input $x\in\{0,1\}^n$, receives two polynomial-size witness registers $A,B$ and has an accepting effect $M_x$ satisfying:
\begin{itemize}
    \item {\bfseries Completeness.} If \(x\in L_{\mathrm{yes}}\), then
    \(
        \max_{\rho\in\mathcal{W}_h(A:B)}
        \Tr(M_x\rho)\ge b(n).
    \)
    \item {\bfseries Soundness.} If \(x\in L_{\mathrm{no}}\), then
    \(
        \max_{\rho\in\mathcal{W}_h(A:B)}
        \Tr(M_x\rho)\le a(n).
    \)
\end{itemize}
\end{definition}

We write $\QMA(2;w_A,w_B,h;a,b)$ when the witness lengths are specified. The notation $\QMA(2;h)$ takes the union over thresholds with $\Delta\ge 1/\operatorname{poly}(n)$, as in \Cref{sec:prelim}.

At $h=0$, the admissible witnesses are exactly the product states, and hence
\begin{equation}\label{eq:trivial}
\QMA(2;0;a,b)=\QMA(2;a,b).    
\end{equation}

\subsection{Structure of the witness set}
\label{subsec:witness-structure}

We characterize $\mathcal{W}_h(A:B)$ through purifications obtained by retaining the local environments. \Cref{prop:purification} states that these purifications preserve the $D=2^h$ equal Schmidt coefficients of the shared EPR state; in other words, the purifications are \emph{flat states.}

\begin{proposition}
\label{prop:purification}
A state $\rho_{AB}$ belongs to $\mathcal{W}_h(A:B)$ if and only if there are local environments $E_A,E_B$ and orthonormal families $\{\ket{u_i}_{AE_A}\}_{i=1}^{D}$ and $\{\ket{v_i}_{BE_B}\}_{i=1}^{D}$ such that
\begin{equation}\label{eq:flat-purification}
\ket{\Omega}_{AE_A:BE_B} =\frac{1}{\sqrt D}\sum_{i=1}^{D} \ket{u_i}_{AE_A}\ket{v_i}_{BE_B}, \qquad \rho_{AB}=\Tr_{E_AE_B}\proj{\Omega}.
\end{equation}
Moreover, $E_A,E_B$ can be chosen with $\dim(E_A)\le D\dim(A)$ and $\dim(E_B)\le D\dim(B)$. Their qubit sizes are polynomially bounded in $n$ whenever $h$ and those of $A,B$ are.
\end{proposition}

\begin{proof}
Let $\rho_{AB}=(\Lambda_A\otimes\Lambda_B)(\proj{\Phi_D})$. Take Stinespring isometries $V_A:R_A\to A\otimes E_A$ and $V_B:R_B\to B\otimes E_B$ for the two channels. Then $\ket{\Omega}=(V_A\otimes V_B)\ket{\Phi_D}$ gives \eqref{eq:flat-purification}, with $\ket{u_i}=V_A\ket{i}$ and $\ket{v_i}=V_B\ket{i}$. Both families are orthonormal because $V_A,V_B$ are isometries.

Conversely, the two orthonormal families define isometries by $V_A\ket{i}=\ket{u_i}$ and $V_B\ket{i}=\ket{v_i}$. Tracing out $E_A,E_B$ gives local channels that prepare $\rho_{AB}$.

Finally, the Stinespring dimension bound allows us to choose \(E_A, E_B\) with $\dim(E_A)\le D\dim(A)$ and $\dim(E_B)\le D\dim(B)$, yielding the stated polynomial bounds on their sizes.
\end{proof}

\begin{corollary}\label{cor:witness-rank-decomposition}
Every state in $\mathcal{W}_h(A:B)$ is a convex combination of pure states of Schmidt rank at most $D$.
\end{corollary}

\begin{proof}
Let $\{K_j\}_j$ and $\{L_k\}_k$ be Kraus operators for the two local channels. Their output can be written using unnormalized vectors as
\[
\rho_{AB}=\sum_{j,k}\proj{\xi_{jk}}, \qquad \ket{\xi_{jk}}=(K_j\otimes L_k)\ket{\Phi_D}.
\]
Each vector $\ket{\xi_{jk}}$ has Schmidt rank at most $D$, since local operations cannot increase Schmidt rank. Normalizing these vectors gives a convex decomposition with weights $\norm{\xi_{jk}}^2$.
\end{proof}

Next we show that if the prepared witness is pure, it must be flat. The following is a specialization of the pure-state conversion criterion from~\cite[Lemma~6 and Corollary~7]{SFK23}.

\begin{proposition}[Pure witnesses]\label{prop:pure-output}
A pure state $\proj{\psi}_{AB}$ belongs to $\mathcal{W}_h(A:B)$ if and only if $\ket{\psi}$ is flat and has Schmidt rank $2^t$ for some integer $0\le t\le h$.
\end{proposition}

\begin{proof}
Suppose $\proj{\psi}_{AB}\in\mathcal{W}_h(A:B)$ is pure. Since its $AB$ marginal is pure, the purification in \eqref{eq:flat-purification} factors as
\[
\ket{\Omega}=\ket{\psi}_{AB}\otimes\ket{\eta}_{E_AE_B}
\]
for some pure state $\ket{\eta}$. Let $\lambda_i$ and $\mu_j$ be the nonzero Schmidt coefficients of $\ket{\psi}$ and $\ket{\eta}$. The \(D\) Schmidt coefficients of $\ket{\Omega}$ across $AE_A:BE_B$ are the products $\lambda_i\mu_j$ and, by \Cref{prop:purification}, are all equal to \(D^{-1/2}\). Holding either index fixed therefore shows that both \(\ket{\psi}\) and \(\ket{\eta}\) are flat. Their Schmidt ranks satisfy $\SR(\psi)\SR(\eta)=D=2^h$, so $\SR(\psi)=2^t$ for some $0\le t\le h$.

Conversely, suppose $\ket{\psi}$ is flat with Schmidt rank $R=2^t$, where $0\le t\le h$, that is,
\[
\ket{\psi}_{AB}=\frac{1}{\sqrt R}\sum_{i=1}^{R}\ket{a_i}_A\ket{b_i}_B.
\]
To generate \(\ket{\psi}\), the provers discard $h-t$ EPR pairs, leaving $\ket{\Phi_R}$ on their remaining registers. Next, using local isometries $V_A\ket{i}=\ket{a_i}$ and $V_B\ket{i}=\ket{b_i}$, they transform $\ket{\Phi_R}$ into $\ket{\psi}$. This preparation uses only local channels, so $\proj{\psi}_{AB}\in\mathcal{W}_h(A:B)$.
\end{proof}

The usual convexity argument for restricting to pure witnesses in $\QMA(2)$ does not directly apply to $\QMA(2;h)$. The normalized states $\ket{\xi_{jk}}/\norm{\xi_{jk}}$ in \Cref{cor:witness-rank-decomposition} may have unequal Schmidt coefficients and therefore need not belong to $\mathcal{W}_h(A:B)$. To work with pure states, we retain the local environments and use the flat purifications of \Cref{prop:purification}.

%% file: sections/03_monotonicity.tex
\section{No cheating advantage from polynomially many EPR pairs}
\label{subsec:monotonicity}

We now show that increasing the number of shared EPR pairs does not reduce the power of the model: $\QMA(2;h)\subseteq\QMA(2;H)$ whenever $h\le H$.
The provers can discard $H-h$ pairs, so $\mathcal{W}_h(A:B)\subseteq\mathcal{W}_H(A:B)$. To establish the result, we must show that the larger witness set gives dishonest provers no additional cheating power against a suitably adapted verifier.

Given a $\QMA(2;h)$ verifier,
\Cref{prot:epr-padding} constructs a $\QMA(2;H)$ verifier.
\Cref{thm:padding}, proved in \Cref{app:exact-padding},
shows that this construction preserves the maximum acceptance
probability on every input.
For an accepting effect $M$ on $AB$, write
\(
    \beta_h(M)
    := \max_{\rho\in\mathcal{W}_h(A:B)} \Tr(M\rho).
\)

\captionsetup[algorithm]{name=Protocol}
\begin{algorithm}[t]
\caption{Testing surplus EPR pairs.}\label{prot:epr-padding}
\begin{algorithmic}[1]
\Require Witnesses on $AC_A$ and $BC_B$, with $(H-h)$-qubit
surplus registers $C_A,C_B$.
\State Perform a Bell measurement on each corresponding
pair of qubits in $C_A C_B$.
\State Reject unless every outcome corresponds to $\ket{\Phi^+}$.
\State Run the original verifier on $AB$
and return its decision.
\end{algorithmic}
\end{algorithm}

\begin{lemma}\label{thm:padding}
Let $h\le H$ be integers with $K=2^{H-h}$, and let surplus EPR registers $C_A,C_B$ each contain $H-h$ qubits. For every accepting effect $0\preceq M\preceq\id_{AB}$,
\begin{equation}\label{eq:padding}
\beta_H\!\left(M_{AB}\otimes\proj{\Phi_K}_{C_AC_B}\right)=\beta_h(M),
\end{equation}
where the left-hand side uses the bipartition $AC_A:BC_B$.
\end{lemma}

By~\Cref{thm:padding}, our simulation in \Cref{prot:epr-padding}
preserves completeness and soundness exactly, yielding the
following precise form of \Cref{res:one}.

\begin{theorem}\label{cor:monotonicity}
For polynomially bounded $h\le H$,
\begin{equation}\label{eq:width-padding}
\QMA(2;w_A,w_B,h;a,b)\subseteq\QMA(2;w_A+H-h,w_B+H-h,H;a,b).
\end{equation}
Consequently, $\QMA(2;h;a,b)\subseteq\QMA(2;H;a,b)$. Then, taking $h=0$ gives
\begin{equation}\label{eq:forward}
\QMA(2)\subseteq\QMA(2;H).
\end{equation}
\end{theorem}
\begin{proof}
Let $L\in\QMA(2;w_A,w_B,h;a,b)$, and let $M_x$ be the
accepting effect of its verifier on input $x$.
Apply \Cref{prot:epr-padding} with $H-h$ surplus EPR pairs.
The resulting verifier has accepting effect
$M_x\otimes\proj{\Phi_K}$, where $K:=2^{H-h}$.
By \Cref{thm:padding},
\[
\beta_H\!\left(M_x\otimes\proj{\Phi_K}\right)=\beta_h(M_x).
\]
Thus, on yes-instances the maximum acceptance probability
remains at least $b$, while on no-instances it remains at
most $a$.

Each witness gains $H-h$ qubits, and the Bell measurements
require $O(H-h)$ additional gates. Since $h,H$ are polynomially bounded, the adapted verifier
runs in polynomial time, proving \eqref{eq:width-padding}.
Taking $h=0$ and recalling \(\QMA(2;0) = \QMA(2)\) then gives
\eqref{eq:forward}.
\end{proof}

Combining \Cref{cor:monotonicity} with a simple padding argument yields the following corollary.

\begin{corollary}[Collapse via polynomial padding]\label{cor:polynomial-budgets}
Fix a rational constant $\varepsilon>0$. Then
\[
\QMA(2; n^\varepsilon)=\QMA(2)
\quad\implies\quad
\QMA(2;h)=\QMA(2)
\]
for every polynomially bounded $h$.
\end{corollary}

Thus, equality for $ h = n^{0.001}$ would settle the question for any \(h = \operatorname{poly}(n)\). We defer the standard but technical proof
to \Cref{app:polynomial-padding}.

%% file: sections/04_removing_log_entanglement.tex
\section{Logarithmic shared entanglement does not affect QMA(2)}\label{sec:removing-entanglement}

We show that logarithmically many shared EPR pairs do not change
the power of $\QMA(2)$.
\Cref{subsec:isometry-encodings} shows how to encode a
$\QMA(2;h)$ witness using two unentangled states;
\Cref{subsec:four-witness-verifier} constructs a $\QMA(4)$
simulation; finally, \Cref{subsec:HM13_tools} reduces it
to a $\QMA(2)$ protocol.

\subsection{Encoding \texorpdfstring{$\QMA(2;h)$}{QMA(2;h)} witnesses}\label{subsec:isometry-encodings}
Let $A$ and $B$ be the witness registers of a $\QMA(2;h)$
protocol, containing $w_A$ and $w_B$ qubits, respectively.
The Stinespring bound in \Cref{prop:purification} allows us
to fix environments $E_A,E_B$ of $w_A+h,w_B+h$ qubits,
sufficient for every local preparation.
Set $D:= 2^h$, $X:=AE_A$ and $Y:=BE_B$.

By \Cref{prop:purification}, a legal witness $\rho_{AB}$ has a flat
purification $\ket{\Omega}=D^{-1/2}\sum_{i=1}^{D}\ket{a_i}_X\ket{b_i}_Y$,
with both families orthonormal. Introduce $h$-qubit index registers $I,J$ and define
\begin{equation}\label{eq:isometry-encodings}
\ket{\alpha}_{XI}
:=\frac{1}{\sqrt D}\sum_{i=1}^{D}\ket{a_i}_X\ket{i}_I,
\qquad
\ket{\beta}_{YJ}
:=\frac{1}{\sqrt D}\sum_{i=1}^{D}\ket{b_i}_Y\ket{i}_J.
\end{equation}
Their index marginals are maximally mixed:
$\alpha_I=\id_I/D$ and $\beta_J=\id_J/D$.
We therefore call $\ket{\alpha}_{XI}$ and $\ket{\beta}_{YJ}$
\emph{proper encodings}; together, they encode $\ket{\Omega}_{XY}$.

A successful Bell projection on $IJ$ reconstructs $\ket{\Omega}_{XY}$ from its proper encoding:
\begin{equation}\label{eq:encoding-bell-contraction}
\bra{\Phi_D}_{IJ}
\left(\ket{\alpha}_{XI}\otimes\ket{\beta}_{YJ}\right)
=\frac{1}{D^{3/2}}\sum_{i=1}^{D}\ket{a_i}_X\ket{b_i}_Y
=\frac{\ket{\Omega}}{D}.
\end{equation}

Thus, the Bell projection succeeds with probability $D^{-2}$
and leaves $\ket{\Omega}$ after normalization. Discarding the environments then recovers the original
$\QMA(2;h)$ witness $\rho_{AB}$.

\subsection{Four-witness simulation}\label{subsec:four-witness-verifier}

We write $\QMA(4;w_1,w_2,w_3,w_4;a,b)$ to make the four
witness sizes explicit. \Cref{thm:four-witness-simulation} gives a
four-witness simulation of $\QMA(2;h)$.

\begin{lemma}\label{thm:four-witness-simulation}
    Let $D:=2^h$ and $\Delta:=b-a>0$. There exist efficiently computable thresholds $0\le a'<b'\le1$ such that
    \begin{equation}\label{eq:four-witness-gap}
    \begin{gathered}
    \QMA(2;w_A,w_B,h;a,b)\subseteq\QMA(4;w'_A,w'_A,w'_B,w'_B;a',b'),\\[3pt]
    w'_A:=2w_A+2h,\qquad w'_B:=2w_B+2h,\qquad b'-a'=\Omega\!\left(\frac{\Delta^2}{D(1+\Delta D)}\right).
    \end{gathered}
    \end{equation}
\end{lemma}

We first describe our $\QMA(4)$ construction, then prove
\Cref{thm:four-witness-simulation} by showing that it achieves
the claimed completeness--soundness gap.

\paragraph*{Construction}

We retain the notation of \Cref{subsec:isometry-encodings}
for the witness registers $A,B$ and encoding registers $XI,YJ$.
Fix a $\QMA(2;h)$ verifier for a promise problem
$L=(L_{\mathrm{yes}},L_{\mathrm{no}})$.
On input $x$, let $M=M_x$ be its accepting effect and set
$\widetilde{M}:=M_{AB}\otimes\id_{E_AE_B}$,
extending $M$ to act trivially on the local environments.

The simulating \(\QMA(4)\) verifier receives four unentangled witnesses, grouped into pairs $(XI,X'I')$ and $(YJ,Y'J')$. Primed registers are copies of the corresponding unprimed registers. Honest provers are expected to submit $(\ket{\alpha},\ket{\alpha}),(\ket{\beta},\ket{\beta})$ using the proper encodings from \eqref{eq:isometry-encodings}. The $\QMA(4)$ verifier in \Cref{prot:four-witness-simulation}
probabilistically combines the Bell test (branches (a) and (b)),
the consistency test (branch (c)), and the properness test
(branch (d)).\footnote{See \Cref{subsec:intro_four_witness}
for an intuitive explanation of these tests.} The branch choice depends on parameters $\lambda$ and $R$, whose values we choose below to obtain the desired completeness--soundness gap.
By convexity, it suffices to analyze our $\QMA(4)$ verifier
on pure witnesses $\ket{\alpha},\ket{\alpha'},\ket{\beta},\ket{\beta'}$.
We analyze the Bell and SWAP contributions separately,
then combine them to obtain the overall acceptance probability.

\begin{algorithm}[t]
\caption{Four-witness simulation of \(\QMA(2;h)\).}\label{prot:four-witness-simulation}
\algrenewcommand{\algorithmicrequire}{\textbf{Registers:}}
\begin{algorithmic}[1]
\Require Witnesses on $XI,X'I',YJ,Y'J'$, with $h$-qubit index
registers $I,I',J,J'$.
\State Choose exactly one of the following branches
according to the stated probabilities.
\end{algorithmic}
\renewcommand{\arraystretch}{1.15}
\noindent\hspace*{\algorithmicindent}%
\begin{tabularx}{\dimexpr\linewidth-\algorithmicindent\relax}{
  @{}l@{\quad}l@{\qquad}X@{}
}
{(a):} & with probability $1/(2R)$ & Perform the Bell test on $IJ$. Reject on failure; otherwise run the original verifier on $AB$ and return its decision. \\[3pt]
{(b):} &  with probability $a/(2R)$ & Perform the Bell test on $IJ$ and accept exactly when it fails. \\[3pt]
{(c):} &  with probability $2\lambda/R$ & Choose either witness pair uniformly. Perform the SWAP test on the full witnesses and accept exactly on outcome $+1$. \\[3pt]
{(d):} &  with probability $2\lambda/R$ & Choose either witness pair uniformly. Perform the SWAP test on its index registers and accept exactly on outcome $-1$. \\[3pt]
{(e):} &  with probability $1/(2R)$ & Accept. \\[3pt]
{(f):} & remaining probability & Reject.
\end{tabularx}
\end{algorithm}

\paragraph*{Bell branches (a)--(b)}
By \eqref{eq:encoding-bell-contraction}, the Bell projection
succeeds with probability $D^{-2}$ on every proper encoding,
whether honest or dishonest. The purpose of these two branches
is to favor honest encodings on yes-instances and penalize
proper encodings on no-instances.

For the unprimed witnesses, let $q$ be the Bell-success probability on $IJ$ and $u$ the joint probability of Bell success on \(IJ\) and acceptance of $\widetilde{M}$ on \(XY\). Including the branch-selection weights, their total acceptance contribution is
\begin{equation}\label{eq:bell-branch-contribution}
\frac{u+a(1-q)}{2R}
=\frac{a + g}{2R},
\qquad
g:=u-aq.
\end{equation}
For proper encodings, a successful Bell test reconstructs a legal
witness, so the completeness and soundness guarantees of the original $\QMA(2;h)$ protocol apply: $u\le aq$ on no-instances, while honest encodings on yes-instances satisfy $u\ge bq$. Thus, every proper encoding satisfies $g\le0$ on no-instances, while honest encodings achieve $g>0$ on yes-instances.


\paragraph*{SWAP branches (c)--(d)} For the pair $\ket{\alpha},\ket{\alpha'}$, the full-witness and index SWAP expectations are $\abs{\ip{\alpha}{\alpha'}}^2$ and $\Tr(\alpha_I\alpha'_{I'})$, respectively. We combine these quantities into the penalty
\begin{equation}\label{eq:index-penalty-definition}
P(\alpha,\alpha'):=1-\abs{\ip{\alpha}{\alpha'}}^2+\Tr(\alpha_I\alpha'_{I'})-\frac1D.
\end{equation}
The term $1-\abs{\ip{\alpha}{\alpha'}}^2$ penalizes disagreement. When the submissions agree, $P$ equals the index purity minus $1/D$. The combined penalty is nonnegative and vanishes for identical proper encodings; \Cref{app:encoding-estimates} gives the quantitative bounds. Put $P_A:=P(\alpha,\alpha')$ and $P_B:=P(\beta,\beta')$.

\paragraph*{Overall acceptance score}

Adding the Bell and SWAP contributions and the unconditional acceptance branch gives
\begin{equation}\label{eq:four-witness-acceptance}
\begin{aligned}
p_{\mathrm{acc}}
&=\underbrace{\frac{a+g}{2R}}_{\text{Bell branches}}
+\underbrace{\frac{\lambda(6-2/D-P_A-P_B)}{2R}}_{\text{SWAP branches}}
+\underbrace{\frac{1}{2R}}_{\text{Branch (e)}}\\
&=\frac12+\frac{g-\lambda(P_A+P_B)}{2R}
\qquad
\bigl(\text{set } R:=1+a+2\lambda(3-1/D)\bigr).
\end{aligned}
\end{equation}
The negative term $-\lambda(P_A+P_B)/(2R)$ comes entirely
from the SWAP branches and vanishes when each witness pair
consists of identical proper encodings. Thus, bounding the
acceptance probability reduces to bounding the score
$S:=g-\lambda(P_A+P_B)$.
\Cref{lem:four-witness-score} gives the required bounds;
its proof is deferred to \Cref{app:simulation-analysis}.

\begin{lemma}[Score bounds]\label{lem:four-witness-score}
For $\lambda:=4+16/(\Delta D)$, the score $S$ satisfies:
\begin{itemize}
    \item If $x\in L_{\mathrm{yes}}$, some choice of four
    unentangled pure witnesses achieves
    \(
        \displaystyle S\ge \frac{\Delta}{D^{2}}.
    \)
    \item If $x\in L_{\mathrm{no}}$, every choice of four
    unentangled pure witnesses satisfies
    \(
        \displaystyle S\le \frac{\Delta}{4D^2}.
    \)
\end{itemize}
\end{lemma}

We are now ready to prove correctness of our \(\QMA(4)\) construction by combining
the score bounds of \Cref{lem:four-witness-score} with the
acceptance identity \eqref{eq:four-witness-acceptance}.

\begin{proof}[Proof of \Cref{thm:four-witness-simulation}]
Rewriting the overall acceptance probability of
\Cref{prot:four-witness-simulation} in terms of the score gives \(p_{\mathrm{acc}}=1/2+S/(2R)\). Thus, \Cref{lem:four-witness-score} yields soundness and completeness thresholds
\begin{equation}
    a_0:=\frac12+\frac{\Delta}{8RD^2},
    \qquad
    b_0:=\frac12+\frac{\Delta}{2RD^2},
    \qquad
    \gamma:=b_0-a_0=\frac{3\Delta}{8RD^2}.
\end{equation}
Approximating the branch probabilities to sufficient precision
preserves a constant fraction of this gap, yielding the gap bound
in \eqref{eq:four-witness-gap}. The technical details are deferred
to \Cref{app:coin_toss}.

The claimed witness sizes follow from the encoding construction
in \Cref{subsec:isometry-encodings}. The protocol runs in polynomial
time, since the Bell measurement and SWAP tests are efficient
and the original $\QMA(2;h)$ verifier is invoked at most once.
\end{proof}

\subsection{Reducing to QMA(2)}\label{subsec:HM13_tools}

Having established a $\QMA(4)$ simulation of $\QMA(2;h)$,
we now apply the Harrow--Montanaro construction~\cite{HM13}
to obtain a $\QMA(2)$ protocol for $h=O(\log n)$.
\begin{theorem}\label{thm:logarithmic-entanglement}
    For every $h=O(\log n)$,
    \begin{equation}\label{eq:logarithmic-entanglement}
    \QMA(2;h)=\QMA(2).
    \end{equation}
\end{theorem}
\begin{proof}
For $h=O(\log n)$, \Cref{thm:four-witness-simulation}
gives a $\QMA(4)$ simulation with an inverse-polynomial
completeness--soundness gap. The Harrow--Montanaro
theorem~\cite[Theorem~9]{HM13} therefore yields a $\QMA(2)$
protocol with completeness at least $2/3$ and soundness
at most $1/3$. Thus, $\QMA(2;h)\subseteq\QMA(2)$.
The reverse inclusion follows from \Cref{cor:monotonicity}.
\end{proof}

%% file: sections/05_LOCC.tex
\section{The LOCC setting}\label{sec:classical}

In this section, we explore a variant of \(\QMA(2;h)\) where, aside from sharing \(h\) EPR pairs, the provers may use classical communication to prepare the witness. Since this setting allows local operations and classical communication (LOCC), we denote the resulting class by \(\QMA^{\mathrm{LOCC}}(2;h)\).

Write $D := 2^h$. The witness set for $\QMA^{\mathrm{LOCC}}(2;h)$ is exactly the set of mixed states of Schmidt number at most $D$. More formally,
\begin{equation}\label{eq:sn-set}
\mathcal{W}^{\mathrm{LOCC}}_h(A:B)=\{\rho\in\mathcal{D}(A\otimes B):\operatorname{SN}(\rho)\le D\}.
\end{equation}
To see this, note that the \(h\) shared EPR pairs $\ket{\Phi_D}$ have Schmidt number $D$ and LOCC cannot increase it. Hence every generated state lies in the set~\eqref{eq:sn-set}. Conversely, consider \(\rho = \sum_j p_j\proj{\psi_j}\) in the set~\eqref{eq:sn-set}, with $\SR(\psi_j)\le D$. Nielsen's majorization theorem~\cite{Nie99} guarantees a deterministic LOCC conversion from $\ket{\Phi_D}$ to each $\ket{\psi_j}$ in the mixture. With probability $p_j$, the provers perform the corresponding conversion. Discarding the label $j$ yields $\rho$~\cite{BD11}. All witnesses in \(\mathcal{W}^{\mathrm{LOCC}}_h\) are convex combinations of pure states of Schmidt rank at most $D$. Our analysis of \(\QMA^{\mathrm{LOCC}}(2;h)\) can thus be restricted to pure states. 

Define $\QMA^{\mathrm{LOCC}}(2;h;a,b)$ by substituting $\mathcal{W}^{\mathrm{LOCC}}_h$ into \Cref{def:qmah} using the same conventions. At $h=0$, the witnesses are separable, giving $\QMA^{\mathrm{LOCC}}(2;0;a,b)=\QMA(2;a,b)$.

\paragraph*{No cheating advantage from polynomially many EPR pairs}

Even in the LOCC setting, additional EPR pairs give dishonest provers no cheating advantage. To show this, we reuse \Cref{prot:epr-padding}. The construction again preserves the maximum acceptance probability, yielding the following analogue of \Cref{cor:monotonicity}.
\begin{theorem}\label{prop:sn-padding}
For polynomially bounded $h\le H$,
\begin{equation}\label{eq:sn-width-padding}
\QMALOCC(2;w_A,w_B,h;a,b)\subseteq\QMALOCC(2;w_A+H-h,w_B+H-h,H;a,b).
\end{equation}
Consequently, $\QMALOCC(2;h;a,b)\subseteq\QMALOCC(2;H;a,b)$. Then, taking $h=0$ gives
    \begin{equation}\label{eq:LOCC_forward}
    \QMA(2)\subseteq\QMALOCC(2;H).
    \end{equation}
\end{theorem}
The full proof is deferred to \Cref{app:locc-proofs}.
The same padding argument extends \Cref{cor:polynomial-budgets}
to the LOCC setting: if \(\QMA^{\mathrm{LOCC}}(2;n^\varepsilon)=\QMA(2)\)
for some fixed rational \(\varepsilon>0\), then
\(\QMA^{\mathrm{LOCC}}(2;h)=\QMA(2)\) for every polynomially bounded \(h\).

\paragraph*{LOCC and logarithmic shared entanglement do not affect QMA(2)}

For \(h=O(\log n)\), we show that \(\QMA^{\mathrm{LOCC}}(2;h)=\QMA(2)\). In the LOCC setting, a direct two-witness simulation is possible
since flatness need not be enforced. This eliminates the
primed copies and SWAP branches of \Cref{prot:four-witness-simulation}. The simulation in \Cref{prot:two-witness-simulation} achieves
the parameters stated in \Cref{thm:sn-simulation}.

\begin{algorithm}[t]
\caption{Two-witness simulation.}\label{prot:two-witness-simulation}
\algrenewcommand{\algorithmicrequire}{\textbf{Registers:}}
\begin{algorithmic}[1]
\Require Witnesses on $AI$ and $BJ$, with
$h$-qubit index registers $I,J$.
\State Perform the Bell test on $IJ$.
\label[step]{step:two-witness-bell}
\State \hspace{\algorithmicindent}\textbf{if} successful,
run the original \(\QMA^{\mathrm{LOCC}}(2;h)\) verifier
on $AB$ and return its decision.
\label[step]{step:two-witness-original}
\State \hspace{\algorithmicindent}\textbf{else}
accept with probability $a$ and reject otherwise.
\label[step]{step:two-witness-coin}
\end{algorithmic}
\end{algorithm}

\begin{theorem}\label{thm:sn-simulation}
Let $\Delta:=b-a>0$. Then,
\begin{equation}\label{eq:sn-thresholds}
\QMA^{\mathrm{LOCC}}(2;w_A,w_B,h;a,b)\subseteq\QMA(2;w_A+h,w_B+h;a',b'),
\end{equation}
where \(a':=a\) and \(b':=a+\frac{\Delta}{2D^2}\).
\end{theorem}

\begin{proof}
The two witnesses occupy \(AI\) and \(BJ\), where \(A,B\)
contain \(w_A,w_B\) qubits, respectively, and \(I,J\) are
\(h\)-qubit index registers. By convexity, it suffices to
analyze the \(\QMA(2)\) verifier in \Cref{prot:two-witness-simulation} on pure product witnesses
\(\ket{\alpha}_{AI},\ket{\beta}_{BJ}\).

\subparagraph{Soundness.} Fix a no-instance. Let \(D:=2^h\) and expand the witnesses
in the computational bases of their index registers:
\[
\ket{\alpha}_{AI}=\sum_{i=1}^{D}\ket{a_i}_A\ket{i}_I,
\qquad
\ket{\beta}_{BJ}=\sum_{i=1}^{D}\ket{b_i}_B\ket{i}_J,
\]
where the vectors \(\ket{a_i},\ket{b_i}\) need not be
normalized or orthogonal. Bell success in
Step~\ref{step:two-witness-bell} leaves \(AB\) in a state
proportional to \(\sum_{i=1}^{D}\ket{a_i}_A\ket{b_i}_B\).
This has Schmidt rank at most \(D\), so the normalized
state is a legal \(\QMALOCC(2;h)\) witness.
The original verifier in Step~\ref{step:two-witness-original}
therefore accepts with probability at most \(a\).
On Bell failure, Step~\ref{step:two-witness-coin} also
accepts with probability at most \(a\).
Thus, the simulation has soundness at most \(a\).

\subparagraph{Completeness.}
On a yes-instance, choose a pure \(\QMALOCC(2;h)\) witness accepted
by the original verifier with probability at least \(b\).
Its Schmidt decomposition is
\[
\ket{\psi}_{AB}
=\sum_{i=1}^{r}\mu_i\ket{a_i}_A\ket{b_i}_B,
\qquad r\le D.
\]
The honest provers submit
\[
\ket{\alpha}_{AI}
=\frac{1}{\sqrt r}\sum_{i=1}^{r}\ket{a_i}_A\ket{i}_I,
\qquad
\ket{\beta}_{BJ}
=\sum_{i=1}^{r}\mu_i\ket{b_i}_B\ket{i}_J.
\]
The Bell projection in Step~\ref{step:two-witness-bell}
produces the unnormalized vector \(\ket{\psi}/\sqrt{Dr}\).
It therefore succeeds with probability \(1/(Dr)\ge D^{-2}\)
and leaves \(\ket{\psi}\) after normalization.
With an exact \(a\)-coin in Step~\ref{step:two-witness-coin},
the acceptance probability satisfies
\begin{equation}\label{eq:locc-simulation-completeness}
p_{\mathrm{acc}}
\ge
\underbrace{\frac{b}{Dr}
\vphantom{a\left(1-\frac1{Dr}\right)}}_{
\text{Step~\ref{step:two-witness-original}}}
+
\underbrace{a\left(1-\frac1{Dr}\right)}_{
\text{Step~\ref{step:two-witness-coin}}}
=a+\frac{\Delta}{Dr}
\ge a+\frac{\Delta}{D^2}.
\end{equation}
Accounting for the finite-precision implementation of the
coin yields the stated completeness and soundness thresholds.
The details are deferred to \Cref{app:LOCC-coin}.

The simulation runs in polynomial time, since the Bell
measurement can be performed efficiently and the original
verifier is invoked at most once.
\end{proof}

%% file: sections/06_logarithmic_witnesses.tex
\section{Logarithmic-size witnesses}\label{sec:log}

$\QMA_{\log}(2)$ is the variant of $\QMA(2)$ in which each witness is restricted to $O(\log n)$ qubits. Applying the same restriction to our models defines $\QMA_{\log}(2;h)$ and $\QMA_{\log}^{\mathrm{LOCC}}(2;h)$, using the parameter conventions of \Cref{def:qmah}. In this section, all classes allow an inverse-polynomial
completeness--soundness gap.

We first extend Theorems~\ref{thm:logarithmic-entanglement} and~\ref{thm:sn-simulation} to logarithmic-size witnesses, showing that sharing $O(\log n)$ EPR pairs does not change the power of either model. For the \(\QMA(2;h)\) case, we use the following logarithmic-size Harrow--Montanaro reduction.

\begin{lemma}[Logarithmic-size Harrow--Montanaro reduction]\label{lem:logarithmic-grouping}
For every fixed integer $k\ge 2$,
\[
    \QMA_{\log}(k)=\QMA_{\log}(2).
\]
\end{lemma}
The proof of \Cref{lem:logarithmic-grouping} is deferred to \Cref{app:witness-grouping}. Combining this reduction with Protocols~\ref{prot:four-witness-simulation} and~\ref{prot:two-witness-simulation}  yields the
following result for logarithmic-size witnesses.

\begin{theorem}\label{thm:log-equality}
For every $h=O(\log n)$,
\begin{equation}\label{eq:log-equalities}
\QMA_{\log}(2;h)=\QMA_{\log}^{\mathrm{LOCC}}(2;h)=\QMA_{\log}(2).
\end{equation}
\end{theorem}

\begin{proof}
For the inclusions from $\QMA_{\log}(2)$, apply
\Cref{cor:monotonicity,prop:sn-padding} starting from zero
shared EPR pairs. Each witness gains $h=O(\log n)$ qubits,
while completeness and soundness are preserved.
For the reverse inclusions, both simulations preserve
logarithmic witness size and an inverse-polynomial
completeness--soundness gap when $h=O(\log n)$.
Specifically, \Cref{thm:four-witness-simulation} gives
$\QMA_{\log}(2;h)\subseteq\QMA_{\log}(4)$, and
\Cref{lem:logarithmic-grouping} gives
$\QMA_{\log}(4)=\QMA_{\log}(2)$.
In the LOCC setting, \Cref{thm:sn-simulation} directly gives
$\QMA_{\log}^{\mathrm{LOCC}}(2;h)\subseteq\QMA_{\log}(2)$.
\end{proof}

\subparagraph{Superlogarithmic EPR pairs}\label{subsec:log-superlogarithmic}
In the LOCC model, sufficiently many shared EPR pairs allow
the provers to prepare any state on the witness registers. Since every state on $w_A$ and $w_B$ qubits has Schmidt number at most $2^{\min\{w_A,w_B\}}$, \eqref{eq:sn-set} gives
\begin{equation}\label{eq:log-saturation}
H\ge\min\{w_A,w_B\}\quad\Longrightarrow\quad
\mathcal{W}^{\mathrm{LOCC}}_H(A:B)=\mathcal{D}(A\otimes B).
\end{equation}
For logarithmic-size witnesses, any superlogarithmic $H=\omega(\log n)$ satisfies this condition for all sufficiently large $n$. The two registers then form a single unrestricted logarithmic-size witness, leading to the following collapse to $\BQP$.

\begin{proposition}\label{prop:log-locc-bqp}
For every polynomially bounded $H=\omega(\log n)$,
\begin{equation}\label{eq:log-locc-bqp}
\QMA_{\log}^{\mathrm{LOCC}}(2;H)=\BQP.
\end{equation}
\end{proposition}
\begin{proof}
Fix a protocol with $w_A,w_B=O(\log n)$. For sufficiently large input size \(n\), \eqref{eq:log-saturation} holds, so its two witness registers form a single unrestricted witness of $w_A+w_B=O(\log n)$ qubits. Marriott--Watrous amplification preserves the size of a single witness, and their logarithmic-witness theorem gives a $\BQP$ algorithm for these inputs~\cite{MW05}. The finitely many shorter inputs can be handled by
hardcoding their answers.
For the reverse inclusion, the verifier can simply ignore its witnesses. Hence \(\BQP \subseteq \QMA_{\log}^{\mathrm{LOCC}}(2;H)\).
\end{proof}

Together with $\NP\subseteq\QMA_{\log}(2)$~\cite{BT12}, this yields the following consequence.
\begin{corollary}\label{cor:log-superlogarithmic-barrier}
Assuming $\NP\nsubseteq\BQP$, for every polynomially bounded $H=\omega(\log n)$,
\[
\QMA_{\log}^{\mathrm{LOCC}}(2;H)
\subsetneq \QMA_{\log}(2).
\]
\end{corollary}

Our construction for testing surplus EPR pairs in
\Cref{prot:epr-padding} is consistent with this conclusion.
By \Cref{prop:sn-padding}, increasing from $h=O(\log n)$
to $H=\omega(\log n)$ adds $H-h=\omega(\log n)$ qubits
to each witness, exceeding the logarithmic size bound.



%% file: appendix/A_preliminaries.tex
\section{Preliminaries and definitions}\label{sec:prelim}

\paragraph*{Quantum information} All Hilbert spaces are finite-dimensional. For a register $A$, let $\llin(A)$ denote its linear operators and $\dd(A)$ its density operators. Unless stated otherwise, a pure state \(\ket{\psi}\) is a unit vector. The Schmidt decomposition of a bipartite pure state is
\begin{equation}\label{eq:schmidt}
 \ket\psi_{AB}=\sum_{i=1}^{r}\mu_i\ket{a_i}_A\ket{b_i}_B,
 \qquad \mu_1\ge\cdots\ge\mu_r>0,
 \qquad \sum_{i=1}^{r}\mu_i^2=1,
\end{equation}
where both \(\{\ket{a_i}\}_i\) and \(\{\ket{b_i}\}_i\) are orthonormal. The number $r$ is the \emph{Schmidt rank}, denoted $\SR(\psi)$. We call \(\ket\psi\) \emph{flat} when all its nonzero coefficients are equal. For a bipartite mixed state $\rho_{AB}$, its \emph{Schmidt number} $\operatorname{SN}(\rho)$ is the smallest $R$ such that $\rho$ is a convex combination of pure states of Schmidt rank at most $R$~\cite{TH00}.

A quantum channel $\Lambda:\mathcal{L}(R)\to\mathcal{L}(A)$
is a completely positive trace-preserving map. Its Kraus and Stinespring
representations take the form
\begin{equation}\label{eq:channel}
 \Lambda(\rho)=\sum_j K_j\rho K_j^\dagger
 =\Tr_E(V\rho V^\dagger),\qquad
 \sum_j K_j^\dagger K_j=\id_R,\qquad
 V=\sum_j K_j\otimes\ket j_E.
\end{equation}
Here $V:R\to A\otimes E$ is an isometry with environment $E$, where $\dim(E)\le\dim(R)\dim(A)$.

The trace distance $T(\rho,\sigma):=\tfrac12\norm{\rho-\sigma}_1$ is convex, contracts under channels and, for a measurement effect \(0 \preceq M \preceq \id\), bounds the outcome probability difference
\begin{equation}\label{eq:measurement-bound}
    \left|\operatorname{Tr}[M(\rho-\sigma)]\right|\le T(\rho,\sigma).
\end{equation}
For pure states,
\begin{equation}\label{eq:trace}
 T(\alpha,\beta):=T(\proj\alpha,\proj\beta)
 =\sqrt{1-\abs{\ip\alpha\beta}^2}.
\end{equation}
For a state on $XI$, we write $\alpha_I:=\Tr_X\proj\alpha$, i.e., the subscript denotes marginal \(I\). We also use the Hilbert--Schmidt norm $\norm Z_2=(\Tr Z^\dagger Z)^{1/2}$.
For $D=2^h$, the state of \(h\) EPR pairs is
\begin{equation}\label{eq:epr}
 \ket{\Phi_D}:=\frac1{\sqrt D}\sum_{i=1}^{D}\ket i\ket i.
\end{equation}
A \emph{Bell test} on two $h$-qubit registers has effects $\proj{\Phi_D}$ and $\id-\proj{\Phi_D}$. We call these outcomes success and failure, respectively.

Let $F$ be the SWAP operator on equal-dimensional registers,
with $F\ket{i,j}=\ket{j,i}$. Then
\begin{equation}
    \operatorname{Tr}[F(\rho\otimes\sigma)]
    = \operatorname{Tr}(\rho\sigma).
\end{equation}
The SWAP test measures the observable $F$, yielding $\pm1$ with
probabilities $[1\pm\operatorname{Tr}(\rho\sigma)]/2$.

\paragraph*{Two-prover verification model} In \(\QMA(2)\) two computationally unbounded provers submit independent quantum witnesses to a uniform polynomial-time quantum verifier. The verifier has access to private ancillas initialized to zero.

For input $x$ of length $n$, let $A$ and $B$ be witness registers containing $w_A(n)$ and $w_B(n)$ qubits, respectively, where $w_A,w_B$ are efficiently computable, polynomially bounded functions. The full witness has form
$\rho_A\otimes\rho_B$, where $\rho_A\in\mathcal{D}(A)$ and
$\rho_B\in\mathcal{D}(B)$. The verifier induces an accepting effect
$0\preceq M_x\preceq \id_{AB}$, with acceptance probability
\(\operatorname{Tr}\!\left[M_x(\rho_A\otimes\rho_B)\right].\)

Soundness $a(n)$ and completeness $b(n)$ are efficiently computable
rationals satisfying $0\le a(n)<b(n)\le1$, with gap $\Delta(n):=b(n)-a(n)$.
We suppress dependence on $n$ when clear.

\begin{definition}
A promise problem $L=(L_{\mathrm{yes}},L_{\mathrm{no}})$ is in
$\QMA(2;a,b)$ if it has a uniform polynomial-time quantum
verifier that, on input $x\in\{0,1\}^n$, receives two
independent polynomial-size witnesses on $A,B$ and has
an accepting effect $M_x$ satisfying:
\begin{itemize}
    \item {\bfseries Completeness.} If \(x\in L_{\mathrm{yes}}\), then
    \(
        \max_{\rho_A,\rho_B}
        \Tr[M_x(\rho_A\otimes\rho_B)]\ge b(n).
    \)
    \item {\bfseries Soundness.} If \(x\in L_{\mathrm{no}}\), then
    \(
        \max_{\rho_A,\rho_B}
        \Tr[M_x(\rho_A\otimes\rho_B)]\le a(n).
    \)
\end{itemize}
Both maxima are over $\rho_A\in\mathcal{D}(A)$ and
$\rho_B\in\mathcal{D}(B)$.
\end{definition}

It suffices to consider pure product witnesses: every
mixed product witness is a convex combination of pure product states, and
acceptance probability is linear in the joint state.

When witness lengths need to be explicit, we write
$\QMA(2;w_A,w_B;a,b)$ for protocols using $w_A(n)$ and $w_B(n)$ witness
qubits, respectively.
For $k$ mutually unentangled witnesses, we write $\QMA(k;a,b)$.
The unqualified notation $\QMA(k)$ takes the union over thresholds
satisfying $\Delta\ge 1/\operatorname{poly}(n)$. By the Harrow--Montanaro
amplification theorem, this agrees with the conventional definition with soundness \(1/3\) and completeness \(2/3\). By the same theorem, \(\QMA(2) = \QMA(k)\) for \(2 \leq k \leq \poly(n)\)~\cite[Theorem~9]{HM13}.

%% file: appendix/B_monotonicity.tex
\section{Testing surplus EPR pairs and polynomial padding}
\label[appendix]{app:monotonicity}
This appendix collects the proofs deferred from
\Cref{subsec:monotonicity}.

\subsection{Testing surplus EPR pairs}
\label[appendix]{app:exact-padding}

We prove \Cref{thm:padding}, which shows that
\Cref{prot:epr-padding} preserves the maximum acceptance
probability exactly. Recall that $\beta_h(M)$ denotes the maximum of $\Tr(M\rho)$ over $\rho\in\mathcal{W}_h(A:B)$. For an integer $R\ge1$, let $\mathcal{F}_R(X:Y)$ denote the set of flat states on $X\otimes Y$ with Schmidt rank $R$. The proof uses the following overlap formula.

\begin{lemma}[Overlap with a flat state]\label{lem:flat-overlap}
Let $\ket{\chi}\in X\otimes Y$ be a possibly unnormalized vector with Schmidt coefficients $\mu_1\ge\mu_2\ge\cdots$, padded with zeros. If $\dim(X),\dim(Y)\ge R$, then
\begin{equation}\label{eq:flat-overlap-monotonicity}
\max_{\ket{\Gamma}\in\mathcal{F}_R(X:Y)}\abs{\ip{\chi}{\Gamma}}=\frac{1}{\sqrt R}\sum_{i=1}^{R}\mu_i.
\end{equation}
\end{lemma}

\begin{proof}
In fixed orthonormal bases, let $\ket{\chi}=\sum_{i,j}C_{ij}\ket{i}_X\ket{j}_Y$ and $\ket{\Gamma}=\sum_{i,j}G_{ij}\ket{i}_X\ket{j}_Y$. Their \emph{coefficient matrices} $C=(C_{ij})$ and $G=(G_{ij})$ have decreasingly ordered singular values $\mu_i$ and $\nu_i$, respectively, equal to their Schmidt coefficients. For $\ket{\Gamma}\in\mathcal{F}_R(X:Y)$, we have $\nu_i=R^{-1/2}$ for $i\le R$ and $\nu_i=0$ otherwise. Von Neumann's trace inequality gives
\[
\abs{\ip{\chi}{\Gamma}}=\abs{\Tr(C^\dagger G)}\le\sum_i\mu_i\nu_i=\frac{1}{\sqrt R}\sum_{i=1}^{R}\mu_i.
\]

The maximum in \eqref{eq:flat-overlap-monotonicity} is attained by choosing the Schmidt families of $\ket{\Gamma}$ to match those of $\ket{\chi}$ for the largest $R$ coefficients. If $\ket{\chi}$ has Schmidt rank less than $R$, complete its Schmidt families to $R$ orthonormal vectors with padded zero Schmidt coefficients.
\end{proof}

\begin{paddinglemmarestated}[restated]\label{lem:pad_lemma_restated}
Let $h\le H$ be integers and set $K:=2^{H-h}$.
Let $C_A,C_B$ each contain $H-h$ qubits.
For every accepting effect $0\preceq M\preceq\id_{AB}$,
\begin{equation}\label{eq:padding-restated}
\beta_H\!\left(M_{AB}\otimes\proj{\Phi_K}_{C_AC_B}\right)
=\beta_h(M),
\end{equation}
where the left-hand side uses the bipartition $AC_A:BC_B$.
\end{paddinglemmarestated}

\begin{proof}
Let $D=2^h$, so $2^H=DK$. Given any $\rho_{AB}\in\mathcal{W}_h(A:B)$, the provers can use $h$ of their shared pairs to prepare $\rho_{AB}$ and send the remaining $H-h$ pairs unchanged on $C_AC_B$. Thus $\rho_{AB}\otimes\proj{\Phi_K}_{C_AC_B}$ belongs to $\mathcal{W}_H(AC_A:BC_B)$ and has acceptance probability $\Tr(M\rho_{AB})$. Maximizing over $\rho_{AB}$ proves the lower bound in \eqref{eq:padding-restated}.

For the upper bound, let $\sigma\in\mathcal{W}_H(AC_A:BC_B)$. By \Cref{prop:purification}, it has a flat purification $\ket{\Omega}$ of Schmidt rank $DK$ across $XC_A:YC_B$, with $X=AE_A$, $Y=BE_B$. Define
\[
\widetilde{M}:=M_{AB}\otimes\id_{E_AE_B},
\qquad
\ket{v}:=\left(\id_{XY}\otimes\bra{\Phi_K}_{C_AC_B}\right)\ket{\Omega}.
\]
Here $\ket{v}$ is the unnormalized vector after Bell success; $\widetilde{M}$ applies the original verifier's accepting effect \(M\), ignoring the environments. The joint acceptance probability of both tests is
\[
\Tr\!\left[\left(M_{AB}\otimes\proj{\Phi_K}_{C_AC_B}\right)\sigma\right]=\bra{v}\widetilde{M}\ket{v}.
\]

Our goal is to show that $\bra{v}\widetilde{M}\ket{v}\le\beta_h(M)$. We will use an overlap comparison to show that $\bra{v}\widetilde{M}\ket{v}\le\bra{\Gamma}\widetilde{M}\ket{\Gamma}$ for some $\ket{\Gamma}\in\mathcal{F}_D(X:Y)$. For every such state, tracing out $E_AE_B$ gives a witness in $\mathcal{W}_h(A:B)$ by \Cref{prop:purification}, so
\begin{equation}\label{eq:flat-witness-bound}
\bra{\Gamma}\widetilde{M}\ket{\Gamma}\le\beta_h(M).
\end{equation}

Let $\ket{\chi}$ be a possibly unnormalized vector on $XY$, with nonincreasing Schmidt coefficients $\mu_i$, padded with zeros. Across $XC_A:YC_B$, the Schmidt coefficients of $\ket{\chi}\otimes\ket{\Phi_K}$ are $\mu_i/\sqrt K$, each repeated $K$ times. Its largest $DK$ coefficients then come from the largest $D$ coefficients of $\ket{\chi}$. Since $\ket{\Omega}\in\mathcal{F}_{DK}(XC_A:YC_B)$,\footnote{\(\SR(\Omega) = DK\) implies $\dim(XC_A),\dim(YC_B)\ge DK$. Since $\dim(C_A)=\dim(C_B)=K$, we also have $\dim(X),\dim(Y)\ge D$. The dimension requirement $\dim(X),\dim(Y)\ge R$ of \Cref{lem:flat-overlap} is satisfied.} applying \Cref{lem:flat-overlap} at ranks $DK$ and $D$ gives
\begin{equation}\label{eq:padding-support}
\begin{alignedat}{2}
\abs{\ip{\chi}{v}}
=\abs{\langle\chi\otimes\Phi_K\mid\Omega\rangle}
&\le\frac{1}{\sqrt{DK}}\sum_{i=1}^{D}\frac{K\mu_i}{\sqrt K}
&\qquad&\text{(\Cref{lem:flat-overlap}, $R=DK$)}\\
&=\frac{1}{\sqrt D}\sum_{i=1}^{D}\mu_i\\
&=\max_{\ket{\Gamma}\in\mathcal{F}_D(X:Y)}\abs{\ip{\chi}{\Gamma}}
&\qquad&\text{(\Cref{lem:flat-overlap}, $R=D$)}.
\end{alignedat}
\end{equation}

Choose $\ket{\chi}=\widetilde{M}\ket{v}$ in \eqref{eq:padding-support}. Positivity $\widetilde{M}\succeq0$ gives $\abs{\ip{\chi}{v}}=\bra{v}\widetilde{M}\ket{v}$, hence
\[
\begin{alignedat}{2}
\bra{v}\widetilde{M}\ket{v}
&\le \max_{\ket{\Gamma}\in\mathcal{F}_D(X:Y)}
     \abs{\bra{v}\widetilde{M}\ket{\Gamma}}
&\qquad&\bigl(\text{by }\eqref{eq:padding-support}\bigr)\\
&\le \sqrt{\bra{v}\widetilde{M}\ket{v}}
     \max_{\ket{\Gamma}\in\mathcal{F}_D(X:Y)}
     \sqrt{\bra{\Gamma}\widetilde{M}\ket{\Gamma}}
&\qquad&\text{(Cauchy--Schwarz)}\\
&\le \sqrt{\bra{v}\widetilde{M}\ket{v}\,\beta_h(M)}
&\qquad&\bigl(\text{by }\eqref{eq:flat-witness-bound}\bigr).
\end{alignedat}
\]
It follows that $\bra{v}\widetilde{M}\ket{v}\le\beta_h(M)$,
proving the upper bound in \eqref{eq:padding-restated}.
\end{proof}

\subsection{Polynomial input padding}
\label[appendix]{app:polynomial-padding}

We give the full padding argument to show \Cref{cor:polynomial-budgets}.

\begingroup
\theoremstyle{plain}
\newtheorem*{budgetrestatement}{Corollary~\ref*{cor:polynomial-budgets}}
\begin{budgetrestatement}[Collapse via polynomial padding, restated]
Fix a rational constant $\varepsilon>0$. Then
\[
\QMA(2;\lceil n^\varepsilon\rceil)=\QMA(2)
\quad\implies\quad
\QMA(2;h)=\QMA(2)
\]
for every polynomially bounded $h$.
\end{budgetrestatement}
\endgroup

\begin{proof}
Fix a rational constant $\varepsilon>0$ and assume $\QMA(2;\lceil n^\varepsilon\rceil)=\QMA(2)$. Let $h$ be any polynomially bounded function, and take a promise problem $L\in\QMA(2;h)$ with completeness $b(n)$, soundness $a(n)$, and inverse-polynomial gap. Choose a strictly increasing integer-valued polynomial $N(n)\ge 2n+1$ such that $\lceil N(n)^\varepsilon\rceil\ge h(n)$ for all $n$. For example, $N(n)=C(n+1)^k$ works for sufficiently large positive integers $C,k$, since $h$ is polynomially bounded and $\varepsilon>0$. For an input $x$ of length $n$, define
\begin{equation}\label{eq:input-padding}
    \operatorname{pad}(x):=1^n0x0^{N(n)-2n-1}.    
\end{equation}
Given $x$, we can compute $\operatorname{pad}(x)$ in polynomial time. The resulting string has length $N(n)$ and allows $x$ to be recovered efficiently. Define the promise problem $L^{\mathrm{pad}}$ by applying this map to the yes- and no-instances of $L$; strings not of the form \eqref{eq:input-padding} are outside the promise.

On a valid padded input of length $m=N(n)$, recover $x$. Running the \(\QMA(2;h)\) verifier for $L$ on $x$ gives a verifier for $L^{\mathrm{pad}}$ using $h(n)$ EPR pairs. Since $h(n)\le\lceil m^\varepsilon\rceil$, \Cref{thm:padding} converts this into a $\QMA(2;\lceil m^\varepsilon\rceil)$ protocol for $L^{\mathrm{pad}}$, preserving completeness $b(n)$ and soundness $a(n)$.

To show $L^{\mathrm{pad}}\in\QMA(2;\lceil m^\varepsilon\rceil)$, completeness and soundness must be stated as functions of the padded input length $m$. Set $a'(m):=a(n)$ and $b'(m):=b(n)$ for $m=N(n)$. Strict monotonicity of $N$ allows efficient computation of these thresholds by binary search, while $n\le m$ ensures that their gap remains inverse polynomial in $m$.\footnote{Formally, an inverse-polynomial gap is also required at unused lengths $m$, outside the image of $N$. At these lengths, we set $a'(m)=1/3$ and $b'(m)=2/3$.} Hence $L^{\mathrm{pad}}\in\QMA(2;\lceil m^\varepsilon\rceil)$.

By the assumed equality, $L^{\mathrm{pad}}\in\QMA(2)$. On input $x$, compute $\operatorname{pad}(x)$ and run this $\QMA(2)$ verifier. Since $N(n)$ is polynomially bounded in $n$, the circuit and witness sizes remain polynomial in $n$, with an inverse-polynomial gap. Hence $L\in\QMA(2)$, proving $\QMA(2;h)\subseteq\QMA(2)$. The reverse inclusion follows from \Cref{cor:monotonicity}.
\end{proof}

%% file: appendix/C_removing_entanglement.tex
\section{Analysis of the four-witness simulation}\label[appendix]{app:four-witness-simulation}

This appendix supplies the details of the four-witness simulation
in \Cref{subsec:four-witness-verifier}.
\Cref{app:encoding-estimates} relates SWAP penalties to deviations
from proper encodings, supporting the score bounds in
\Cref{app:simulation-analysis}.
\Cref{app:coin_toss} covers finite-precision branch sampling.

\subsection{Estimates for proper encodings}\label[appendix]{app:encoding-estimates}
Recall that an encoding $\ket{\alpha}_{XI}$ is proper if the index marginal $\alpha_I$ is maximally mixed. We show that the penalty $P$ as defined in
\eqref{eq:index-penalty-definition} controls excess index purity
even when the submitted witnesses $\ket{\alpha}_{XI}$ and
$\ket{\alpha'}_{X'I'}$ differ.

\begin{lemma}[Consistency and index purity]\label{lem:consistency-purity}
For pure $\ket{\alpha}_{XI},\ket{\alpha'}_{X'I'}$ with $\dim(I)=D$,
\begin{equation}\label{eq:consistency-purity}
\begin{aligned}
P(\alpha,\alpha')&:=1-\abs{\ip{\alpha}{\alpha'}}^2+\Tr(\alpha_I\alpha'_{I'})-\frac1D\\
&\ge\frac12\left(\Tr\Big(\alpha_I^2\Big)-\frac1D+\Tr\Big({\alpha'_{I'}}^2\Big)-\frac1D\right)\ge0.
\end{aligned}
\end{equation}
\end{lemma}

\begin{proof}
For the traceless Hermitian operator $Z = \alpha_I - \alpha'_{I'}$, its positive and negative parts have equal trace, so $\norm{Z}_2^2\le\norm{Z}_1^2/2$. Contractivity of trace distance therefore gives
\[
\norm{\alpha_I-\alpha'_{I'}}_2^2
\le\frac12\norm{\alpha_I-\alpha'_{I'}}_1^2
\le2T(\alpha,\alpha')^2
=2\left(1-\abs{\ip{\alpha}{\alpha'}}^2\right).
\]
Substituting this bound into the following expansion of the squared Hilbert--Schmidt norm
\[
\Tr(\alpha_I\alpha'_{I'})=\frac{\Tr\Big(\alpha_I^2\Big)+\Tr\Big({\alpha'_{I'}}^2\Big)-\norm{\alpha_I-\alpha'_{I'}}_2^2}{2}
\]
proves the first inequality in \eqref{eq:consistency-purity}. The second follows because every $D$-dimensional density operator has purity at least $1/D$. Equality is attained by identical proper encodings: the squared overlap is one and the index purity is $1/D$, so $P=0$.
\end{proof}

The next lemma converts a small excess in index purity into proximity to a proper encoding.

\begin{lemma}[Approximation by a proper encoding]\label{lem:encoding-rounding}
Suppose $\dim(X)\ge D$ and $\dim(I)=D$. For a pure state $\ket{\alpha}$ on $XI$, let $e:=\Tr(\alpha_I^2)-1/D$ denote its index purity gap. There is a proper encoding $\ket{\widehat{\alpha}}$ on the same registers satisfying
\begin{equation}\label{eq:encoding-rounding}
\widehat{\alpha}_I=\frac{\id_I}{D},\qquad
T(\alpha,\widehat{\alpha})\le\sqrt{De}.
\end{equation}
\end{lemma}

\begin{proof}
Pad the Schmidt decomposition of $\ket{\alpha}$ to $D$ terms with zeros and define \(\ket{\widehat{\alpha}}\) as
\[
\ket{\alpha}=\sum_{i=1}^{D}\sqrt{p_i}\ket{u_i}_X\ket{v_i}_I,\qquad
\ket{\widehat{\alpha}}=\frac{1}{\sqrt D}\sum_{i=1}^{D}\ket{u_i}_X\ket{v_i}_I.
\]
The condition $\dim(X)\ge D$ allows this completion. The resulting encoding $\ket{\widehat{\alpha}}$ is proper because $\widehat{\alpha}_I=\id_I/D$. Then, using $\ip{\alpha}{\widehat{\alpha}}=D^{-1/2}\sum_i\sqrt{p_i}\in[0,1]$, we obtain
\begin{alignat*}{2}
T(\alpha,\widehat{\alpha})^2
&=1-\abs{\ip{\alpha}{\widehat{\alpha}}}^2\\
&\le2\left(1-\ip{\alpha}{\widehat{\alpha}}\right)\\
{}&=\sum_{i=1}^{D}\left(\sqrt{p_i}-D^{-1/2}\right)^2\\
&\le D\sum_{i=1}^{D}\left(p_i-D^{-1}\right)^2\\
&=D\left(\sum_{i=1}^{D}p_i^2-\frac1D\right)=De,
\end{alignat*}
where the second inequality follows termwise from $(\sqrt{p_i}+D^{-1/2})^2\ge1/D$.
\end{proof}

\subsection{Proof of the score bounds}\label[appendix]{app:simulation-analysis}
For the unprimed witnesses, $q$ denotes the Bell-success
probability on $IJ$, and $u$ the joint probability of Bell
success on \(IJ\) and acceptance by the original verifier on \(AB\). The overall acceptance score is
\[
S:=g-\lambda(P_A+P_B),
\qquad g:=u-aq,
\]
where $P_A:=P(\alpha,\alpha')$ and $P_B:=P(\beta,\beta')$
are the SWAP penalties for the two witness pairs.
We prove the following bounds on the score \(S\).

\begin{scoreboundsrestated}[Score bounds, restated]
Let $D:=2^h$, $\Delta:=b-a>0$, and
$\lambda:=4+16/(\Delta D)$.
\begin{itemize}
    \item If $x\in L_{\mathrm{yes}}$, some choice of four
    unentangled pure witnesses achieves
    $\displaystyle S\ge\frac{\Delta}{D^2}$.
    \item If $x\in L_{\mathrm{no}}$, every choice of four
    unentangled pure witnesses satisfies
    $\displaystyle S\le\frac{\Delta}{4D^2}$.
\end{itemize}
\end{scoreboundsrestated}

\begin{proof}
We retain the notation of \Cref{subsec:four-witness-verifier}.

\subparagraph{Completeness}
On a yes-instance, choose a legal witness accepted with probability
at least $b$. Submit the proper encodings of its flat purification
from \eqref{eq:isometry-encodings}, together with identical primed
copies. Both SWAP penalties vanish, so $P_A=P_B=0$.
By \eqref{eq:encoding-bell-contraction}, $q=D^{-2}$ and
$u\ge bD^{-2}$. Hence,
\[
S=g=u-aq\ge\frac{b-a}{D^2}=\frac{\Delta}{D^2}.
\]

\subparagraph{Soundness} Fix a no-instance. As observed in the analysis of the Bell branches in
\Cref{subsec:four-witness-verifier}, proper encodings satisfy $g\le0$. We bound the possible increase in $g$ for arbitrary (possibly improper) encodings, then subtract the SWAP penalties
$\lambda(P_A+P_B)$ to obtain the desired bound $S\le\Delta/(4D^2)$.

For arbitrary witnesses $\ket{\alpha}_{XI}$ and
$\ket{\beta}_{YJ}$, define their excess index purities by
\[
e_A:=\Tr(\alpha_I^2)-\frac1D,
\qquad
e_B:=\Tr(\beta_J^2)-\frac1D.
\]
Since $\dim (X),\dim (Y)\ge D$, \Cref{lem:encoding-rounding} gives proper encodings
$\ket{\widehat{\alpha}}_{XI}$ and $\ket{\widehat{\beta}}_{YJ}$
satisfying
\[
d_A:=T(\alpha,\widehat{\alpha})\le\sqrt{De_A},
\qquad
d_B:=T(\beta,\widehat{\beta})\le\sqrt{De_B}.
\]
Write $g(\rho,\sigma)$ for the value of $g$ on
witnesses $\rho_{XI}$ and $\sigma_{YJ}$. Since $0\le u\le q$
and $0\le a\le1$, any such density operators satisfy
\begin{equation}\label{eq:g-general-bound}
|g(\rho,\sigma)|
\le q(\rho,\sigma)
=\frac1D\Tr(\rho_I^{\mathsf T}\sigma_J)
\le\frac1D,
\end{equation}
where the transpose is taken in the fixed computational basis.
If either index marginal is maximally mixed, this gives
the stronger bound $|g(\rho,\sigma)|\le D^{-2}$. To obtain the desired soundness, we refine these bounds using the distances $d_A,d_B$ between the witnesses and their proper approximations.

Since $u$ and $q$ are bilinear in the two input density
operators, so is $g$; we extend it bilinearly to Hermitian
inputs. Identifying pure states with their density operators,
set
\[
\delta_A:=\alpha-\widehat{\alpha},
\qquad
\delta_B:=\beta-\widehat{\beta}.
\]
Then
\begin{equation}\label{eq:g-expansion}
g(\alpha,\beta)
=g(\widehat{\alpha},\widehat{\beta})
  +g(\delta_A,\widehat{\beta})
+g(\widehat{\alpha},\delta_B)
  +g(\delta_A,\delta_B).
\end{equation}
The first term is nonpositive because both approximating
encodings \(\ket{\widehat{\alpha}}\) and \(\ket{\widehat{\beta}}\) are proper. It remains to bound the other three terms.

The operator $\delta_A$ is traceless and Hermitian, with
$\|\delta_A\|_1=2d_A$. It therefore admits a decomposition
\[
\delta_A=d_A(\rho_A^+-\rho_A^-),
\]
where $\rho_A^\pm$ are density operators. Since $\widehat{\beta}$ has a maximally mixed index marginal,
we obtain
\[
\begin{alignedat}{2}
|g(\delta_A,\widehat{\beta})|
&=d_A\left|g(\rho_A^+,\widehat{\beta})
          -g(\rho_A^-,\widehat{\beta})\right|
&\qquad&\text{(bilinearity)}\\
&\le d_A\left(|g(\rho_A^+,\widehat{\beta})|
             +|g(\rho_A^-,\widehat{\beta})|\right)
&\qquad&\text{(triangle inequality)}\\
&\le \frac{2d_A}{D^2}
&\qquad&\bigl(\text{by }\eqref{eq:g-general-bound},
\ \text{and }\widehat{\beta}_J=\id_J/D\bigr).
\end{alignedat}
\]
The analogous argument gives
$|g(\widehat{\alpha},\delta_B)|\le 2d_B/D^2$.
Similarly, expanding both differences in $g(\delta_A,\delta_B)$
gives four terms, each bounded in absolute value by $D^{-1}$
using \eqref{eq:g-general-bound}, and hence
\[
|g(\delta_A,\delta_B)|\le\frac{4d_A d_B}{D}.
\]
Substituting these bounds into \eqref{eq:g-expansion} and using
$g(\widehat{\alpha},\widehat{\beta})\le0$, we obtain
\[
\begin{aligned}
g(\alpha,\beta)
&\le \frac{2(d_A+d_B)}{D^2}+\frac{4d_A d_B}{D}\\
&\le \frac{2(\sqrt{e_A}+\sqrt{e_B})}{D^{3/2}}
     +4\sqrt{e_Ae_B}.
\end{aligned}
\]

We now turn to the SWAP penalties, which control the excess
index purities $e_A,e_B$ appearing in the bound above. Discarding the primed excess purity terms in \Cref{lem:consistency-purity} gives
$P_A\ge e_A/2$ and $P_B\ge e_B/2$.
Using $4\sqrt{e_Ae_B}\le2(e_A+e_B)$, we obtain
\[
\begin{aligned}
S
&=g-\lambda(P_A+P_B)\\
&\le \frac{2(\sqrt{e_A}+\sqrt{e_B})}{D^{3/2}}
     -\left(\frac{\lambda}{2}-2\right)(e_A+e_B)\\
&\le \frac{2}{D^3(\lambda/2-2)}\\
&=\frac{\Delta}{4D^2}.
\end{aligned}
\]
The penultimate inequality applies
$c\sqrt{t}-kt\le c^2/(4k)$ separately to $e_A$ and $e_B$,
with $c=2D^{-3/2}$ and $k=\lambda/2-2>0$.
The last equality uses our choice $\lambda=4+16/(\Delta D)$ to bound the soundness score at one quarter of the completeness lower bound $\Delta/D^2$.
\end{proof}

\subsection{Branch sampling}\label{app:coin_toss}
Recall the ideal thresholds $a_0,b_0$ from the proof of
\Cref{thm:four-witness-simulation}, with
$\gamma:=b_0-a_0=3\Delta/(8RD^2)$.

Write $p_1,\ldots,p_6$ for the branch probabilities in \Cref{prot:four-witness-simulation}. Choose the smallest nonnegative integer $t$ with $2^{-t}\le\gamma/20$, and set
\begin{equation}\label{eq:dyadic-branch-probabilities}
\widetilde{p}_j:=2^{-t}\lfloor2^t p_j\rfloor\quad(1\le j\le5),\qquad
\widetilde{p}_6:=1-\sum_{j=1}^{5}\widetilde{p}_j.
\end{equation}
This rounds all branch probabilities to \(t\) fractional bits and yields a rounded probability distribution \(\widetilde{p}_1, \dots, \widetilde{p}_6\) whose total variation distance from the exact distribution satisfies
\[
\frac12\sum_{j=1}^{6}\abs{\widetilde{p}_j-p_j}
=\sum_{j=1}^{5}(p_j-\widetilde{p}_j)
<5\cdot2^{-t}\le\frac{\gamma}{4}.
\]
Every acceptance probability therefore changes by at most
$\gamma/4$. Thus, the implemented verifier has soundness at
most $a'$ and completeness at least $b'$, where
\begin{equation}\label{eq:implemented-four-witness-thresholds}
a':=a_0+\frac{\gamma}{4}
   =\frac12+\frac{7\Delta}{32RD^2}, \qquad
b':=b_0-\frac{\gamma}{4}
   =\frac12+\frac{13\Delta}{32RD^2}.
\end{equation}
These thresholds satisfy $0\le a'<b'\le1$.
Since $R\le7\lambda=O(1+1/(\Delta D))$, their gap satisfies
\[
b'-a'
=\frac{\gamma}{2}
=\frac{3\Delta}{16RD^2}
=\Omega\!\left(\frac{\Delta^2}{D(1+\Delta D)}\right),
\]
as claimed in \eqref{eq:four-witness-gap}.

The rounded branch distribution is sampled efficiently using $t=O(1+h+\log(1/\Delta))$ fair bits. Its weights and the thresholds $a',b'$ are efficiently computable with bit lengths polynomial in $h$ and the bit lengths of $a,b$.

%% file: appendix/D_locc.tex

\section{Proofs for LOCC witness preparation}\label[appendix]{app:locc-proofs}
This appendix collects the proofs deferred from \Cref{sec:classical}.

\subsection{Testing surplus EPR pairs in LOCC}\label[appendix]{app:sn-padding}
We give the full proof that, for \(h\le H\),
\(\QMALOCC(2;h)\subseteq\QMALOCC(2;H)\).
Following the same strategy as in the \(\QMA(2;h)\) setting,
we first prove the LOCC analogue of
\Cref{thm:padding}.
For an accepting effect $M$ on $AB$, write $\beta_h^{\mathrm{LOCC}}(M):=\max_{\rho\in\mathcal{W}^{\mathrm{LOCC}}_h(A:B)}\Tr(M\rho)$.

\begin{lemma}\label{lem:sn-padding}
    Let $h\le H$ be integers, set $K:=2^{H-h}$, and let surplus EPR registers $C_A,C_B$ each contain $H-h$ qubits. For every accepting effect $0\preceq M\preceq\id_{AB}$,
\begin{equation}\label{eq:sn-padding}
\beta_H^{\mathrm{LOCC}}\!\left(M_{AB}\otimes\proj{\Phi_K}_{C_AC_B}\right)=\beta_h^{\mathrm{LOCC}}(M),
\end{equation}
where the left-hand side uses the bipartition $AC_A:BC_B$.
\end{lemma}

\begin{proof}
Put $D:=2^h$, so $2^H=DK$. Tensoring any legal \(\QMA^{\mathrm{LOCC}}(2;h)\) witness with $\proj{\Phi_K}$ gives a witness of Schmidt number at most $DK$ with the same acceptance probability. This proves the lower bound in \eqref{eq:sn-padding}.

For the upper bound, we use the following analogue of
\Cref{lem:flat-overlap} for states of bounded Schmidt rank. If a possibly unnormalized bipartite vector $\ket{\chi}$ has nonincreasing Schmidt coefficients $\mu_i$, padded with zeros, then for every integer $R\ge1$,
\begin{equation}\label{eq:sn-rank-overlap}
\max_{\substack{\norm{\Gamma}=1\\\SR(\Gamma)\le R}}\abs{\ip{\chi}{\Gamma}}=\left(\sum_{i=1}^{R}\mu_i^2\right)^{1/2},
\end{equation}
where $\ket{\Gamma}$ occupies the same bipartite registers as $\ket{\chi}$. Indeed, the von Neumann trace inequality used in the proof of \Cref{lem:flat-overlap} bounds the overlap by $\sum_{i=1}^{R}\mu_i\nu_i$, where $\nu_i$ are the Schmidt coefficients of $\ket{\Gamma}$. Cauchy--Schwarz and $\sum_i\nu_i^2=1$ give the stated bound. The maximum is attained by taking $\ket{\Gamma}$ as the normalized truncation of $\ket{\chi}$ to its largest $R$ Schmidt coefficients; for $\ket{\chi}=0$, both sides vanish.

To prove the upper bound, it suffices by convexity to consider a unit vector $\ket{\Omega}$ of Schmidt rank at most $DK$ across $AC_A:BC_B$. Its unnormalized Bell-success vector and joint acceptance probability with the original accepting effect $M_{AB}$ are
\begin{equation}\label{eq:locc-padding-contraction}
\ket{v}:=(\id_{AB}\otimes\bra{\Phi_K}_{C_AC_B})\ket{\Omega},\qquad
\bra{\Omega}(M_{AB}\otimes\proj{\Phi_K}_{C_AC_B})\ket{\Omega}=\bra{v}M\ket{v}.
\end{equation}
For any vector $\ket{\chi}$ on $AB$ with coefficients $\mu_i$ as above, the Schmidt coefficients of $\ket{\chi}\otimes\ket{\Phi_K}$ are $\mu_i/\sqrt K$, each repeated $K$ times. Applying \eqref{eq:sn-rank-overlap} at ranks $DK$ and $D$ gives
\begin{equation}\label{eq:locc-padding-overlap}
\abs{\ip{\chi}{v}}
=\abs{\langle\chi\otimes\Phi_K\mid\Omega\rangle}
\le\left(\sum_{i=1}^{D}K\frac{\mu_i^2}{K}\right)^{1/2}
=\max_{\substack{\norm{\Gamma}=1\\\SR(\Gamma)\le D}}\abs{\ip{\chi}{\Gamma}}.
\end{equation}
Every $\ket{\Gamma}$ in this maximum is a legal \(\QMA^{\mathrm{LOCC}}(2;h)\) witness. As in the concluding step of the proof of \Cref{thm:padding}, choose $\ket{\chi}=M\ket{v}$ and use positivity and Cauchy--Schwarz:
\[
\bra{v}M\ket{v}
\le\max_{\substack{\norm{\Gamma}=1\\\SR(\Gamma)\le D}}\abs{\bra{v}M\ket{\Gamma}}
\le\sqrt{\bra{v}M\ket{v}\,\beta_h^{\mathrm{LOCC}}(M)}.
\]
If $\bra{v}M\ket{v}=0$, the upper bound is immediate; otherwise divide by its square root and square. Convexity then proves \eqref{eq:sn-padding} for every legal mixed witness.
\end{proof}

The preceding lemma shows that \Cref{prot:epr-padding}
preserves completeness and soundness exactly, yielding
the following theorem, restated from \Cref{sec:classical}.

\begin{locctheoremrestated}[restated]
For polynomially bounded $h\le H$,
\begin{equation}\label{eq:sn-width-padding_restated}
\QMALOCC(2;w_A,w_B,h;a,b)\subseteq\QMALOCC(2;w_A+H-h,w_B+H-h,H;a,b).
\end{equation}
Consequently, $\QMALOCC(2;h;a,b)\subseteq\QMALOCC(2;H;a,b)$. Then, taking $h=0$ gives
    \begin{equation}\label{eq:forward_restated}
    \QMA(2)\subseteq\QMALOCC(2;H).
    \end{equation}
\end{locctheoremrestated}
\begin{proof}
The proof of \Cref{cor:monotonicity} applies unchanged
in the LOCC setting, using the preceding \Cref{lem:sn-padding} in place
of \Cref{thm:padding}.
\end{proof}

\subsection{Coin sampling.}\label{app:LOCC-coin}
We apply the rounding method from \Cref{app:coin_toss} to the \(a\)-biased coin. Choose the smallest nonnegative integer $t$ with $2^{-t}\le\Delta/(2D^2)$ and replace $a$ by \(\widetilde a:=2^{-t}\lfloor2^t a\rfloor\).

Since $\widetilde a\le a$, soundness remains at most $a$. For completeness, only acceptance on Bell failure changes, reducing the overall acceptance probability by at most $a-\widetilde a$. Thus \eqref{eq:locc-simulation-completeness} gives
\[
\widetilde p_{\mathrm{acc}}
\ge a+\frac{\Delta}{D^2}-(a-\widetilde a)
\ge a+\frac{\Delta}{2D^2}.
\]
This proves the thresholds in \eqref{eq:sn-thresholds}.

%% file: appendix/E_witness_grouping.tex
\section{Logarithmic-size Harrow--Montanaro reduction}\label[appendix]{app:witness-grouping}
This appendix proves \Cref{lem:logarithmic-grouping}, which we restate for convenience. We adjust the branch probabilities in the Harrow--Montanaro construction~\cite[Protocol~2]{HM13} to preserve logarithmic witness size and an inverse-polynomial gap.

\begingroup
\theoremstyle{plain}
\newtheorem*{logarithmicgroupingrestatement}{Lemma~\ref{lem:logarithmic-grouping}}
\begin{logarithmicgroupingrestatement}[Logarithmic-size Harrow–Montanaro reduction, restated]
For every fixed integer $k\ge2$,
\[
\QMA_{\log}(k)=\QMA_{\log}(2).
\]
\end{logarithmicgroupingrestatement}
\endgroup

\begin{proof}[Proof of \Cref{lem:logarithmic-grouping}]
Fix a $\QMA_{\log}(k)$ protocol with witness sizes $w_1,\ldots,w_k=O(\log n)$, completeness $b$, soundness $a$, and gap $\Delta:=b-a\ge1/\operatorname{poly}(n)$. We construct a two-witness protocol with $W:=\sum_{i=1}^k w_i$ qubits per witness and gap at least $\Delta^2/256$.

\emph{Construction.} Each prover sends a group containing the $k$ original witness registers. Choose $\tau=2^{-t}$ with $\Delta/128<\tau\le\Delta/64$. With probability $1-\tau$, perform SWAP tests on all corresponding register pairs and accept exactly when every outcome is $+1$. With probability $\tau$, choose either group uniformly and run the original verifier on it.

\emph{Completeness.} On a yes-instance, honest provers send identical copies of a pure product witness accepted by the original verifier with probability at least $b$. All SWAP tests accept with certainty, giving completeness at least
\[
b':=1-\tau(1-b).
\]

\emph{Soundness.} Fix a no-instance. By the usual convexity argument for $\QMA(2)$, it suffices to consider pure independent groups $\ket{\psi_1},\ket{\psi_2}$. Let $\epsilon_j:=1-\max_{\ket{z_1},\ldots,\ket{z_k}}\abs{\langle\psi_j\mid z_1\otimes\cdots\otimes z_k\rangle}^2$, where the $\ket{z_i}$ are unit vectors on the original registers, and put $\epsilon:=(\epsilon_1+\epsilon_2)/2$. Write $p_{\mathrm{test}}(\psi_1,\psi_2)$ for the probability that all SWAP tests accept and $\sigma_{j,S}$ for the marginal of $\proj{\psi_j}$ on registers indexed by $S\subseteq\{1,\ldots,k\}$. The product-test formula and bound~\cite[Lemma~2 and Theorem~1]{HM13} give
\begin{equation}\label{eq:grouping-product-test}
\begin{aligned}
p_{\mathrm{test}}(\psi_1,\psi_2)
&=2^{-k}\sum_{S\subseteq\{1,\ldots,k\}}\Tr(\sigma_{1,S}\sigma_{2,S})\\
&\le\frac{p_{\mathrm{test}}(\psi_1,\psi_1)+p_{\mathrm{test}}(\psi_2,\psi_2)}2
\le1-\frac{\epsilon}{64}.
\end{aligned}
\end{equation}
The first inequality uses $2\Tr(UV)\le\Tr(U^2)+\Tr(V^2)$ for Hermitian $U,V$; the last relaxes the product-test bound $p_{\mathrm{test}}(\psi_j,\psi_j)\le1-11\epsilon_j/512$.

By definition of \(\epsilon_j\) and~\eqref{eq:trace}, each $\ket{\psi_j}$ is at trace distance $\sqrt{\epsilon_j}$ from a product state. By soundness and~\eqref{eq:measurement-bound}, the original verifier accepts each group $\ket{\psi_j}$ with probability at most $a+\sqrt{\epsilon_j}$. Combining the branches yields
\begin{equation}\label{eq:grouping-soundness}
\begin{alignedat}{2}
p_{\mathrm{acc}}
&\le \overbrace{(1-\tau)\left(1-\frac{\epsilon}{64}\right)}^{\text{SWAP branch}}
+\overbrace{\tau\left(a+\sqrt{\epsilon}\right)
\vphantom{\left(\frac{\epsilon}{64}\right)}}^{\text{verification branch}}
&\qquad&\text{(branch bounds and concavity)}\\
&=1-\tau(1-a)+\tau\sqrt{\epsilon}
  -\frac{1-\tau}{64}\epsilon\\
&\le1-\tau(1-a)+\frac{16\tau^2}{1-\tau}
&\qquad&\text{(completing the square)}\\
&\le1-\tau(1-a)+\frac{\tau\Delta}{2}=:a'
&\qquad&\bigl(\tau\le\Delta/64,\ \tau\le1/2\bigr).
\end{alignedat}
\end{equation}
Hence $b'-a'=\tau\Delta/2>\Delta^2/256$.

\emph{Witness size and efficiency.} Since $k$ is fixed, each group has $W=O(\log n)$ qubits. The verifier uses efficient SWAP tests and invokes the original verifier at most once. The branch choice can be sampled exactly using $t=O(1+\log(1/\Delta))$ fair bits, with one additional bit to select a group. The thresholds $a',b'$ are efficiently computable, and their gap remains inverse polynomial. Thus $\QMA_{\log}(k)\subseteq\QMA_{\log}(2)$. Conversely, a $\QMA_{\log}(k)$ verifier can ignore the additional $k-2$ witnesses, proving the reverse containment.
\end{proof}